\documentclass[a4, 10 pt, conference]{ieeetran}    
\IEEEoverridecommandlockouts 

\usepackage{amsmath,amssymb,amsfonts,bm}

\usepackage{amsthm}

\usepackage{booktabs}

\usepackage{mathtools}
\usepackage{graphicx}
\usepackage{pgfplots}
\usepackage[font=small]{subcaption}
\usepackage[font=small]{caption}

\usepackage{tabularx}
\usepackage[table]{xcolor}

\usepackage{colortbl}  

\newtheorem{thm}{Theorem}
\newtheorem{lem}{Lemma}
\newtheorem{cor}{Corollary}
\newtheorem{defn}{Definition}
\newtheorem{prop}{Proposition}
\newtheorem{rem}{Remark}

\newtheorem{asmpt}{Assumption}

\usepackage{enumitem}
\setlist[enumerate]{leftmargin=*}
\setlist[itemize]{leftmargin=*}

\usepackage[
    style=numeric,
    sorting=none,        
    sortcites=true,      
    backend=biber, 
    maxbibnames=99,      
    natbib=true,
    doi=false,   
    url=false,   
    isbn=false,  
    eprint=false, 
    giveninits=true
]{biblatex}

\title{\Huge Global Geometry of Orthogonal Foliations of Signed-Quadratic Systems}

\author{Antonio Franchi$^{1,2}$
\thanks{$^1$Robotics and Mechatronics Department, Electrical Engineering,  Mathematics, and Computer Science Faculty, University of Twente, The Netherlands. 
}
\thanks{$^2$Department of Computer, Control and Management Engineering, Sapienza University of Rome, 00185 Rome, Italy. {\tt\footnotesize schol@r-franchi.eu}}
\thanks{The work was partially funded by the European Commission Horizon Europe Framework under project Autoassess (101120732)}
}

\definecolor{arxivcolor}{rgb}{0.0, 0.15, 0.4}   

\newif\ifshortversion

\shortversionfalse 

\newcommand{\arxivfigure}[1]{%
  \ifshortversion
    \par\vspace{1ex}\noindent
    \textit{For a visual illustration of the results presented in the previous part please see the corresponding figure in the extended version~\cite{my_arxiv_citation}.}\par\vspace{1ex}
  \else
    {
      \let\oldfigure\figure
      \let\endoldfigure\endfigure
      \renewenvironment{figure}[1][htbp]{%
        \oldfigure[##1]\color{arxivcolor}%
        \captionsetup{labelfont={color=arxivcolor}, textfont={color=arxivcolor}}%
      }{\endoldfigure}%

      \expandafter\let\expandafter\oldfigurestar\csname figure*\endcsname
      \expandafter\let\expandafter\endoldfigurestar\csname endfigure*\endcsname
      \renewenvironment{figure*}[1][htbp]{%
        \oldfigurestar[##1]\color{arxivcolor}%
        \captionsetup{labelfont={color=arxivcolor}, textfont={color=arxivcolor}}%
      }{\endoldfigurestar}%
      
      \let\oldtable\table
      \let\endoldtable\endtable
      \renewenvironment{table}[1][htbp]{%
        \oldtable[##1]\color{arxivcolor}%
        \captionsetup{labelfont={color=arxivcolor}, textfont={color=arxivcolor}}%
        \arrayrulecolor{arxivcolor}
      }{\endoldtable}%
      
      {\color{arxivcolor} #1}
    }
  \fi
}

\newif\ifdraftversion
\draftversionfalse 

\ifdraftversion
    \usepackage{fancyhdr}
    \usepackage{xcolor}
    
    \fancypagestyle{plain}{
        \fancyhf{} 
        \fancyfoot[C]{\textcolor{red}{\textbf{DRAFT - CONFIDENTIAL}}}
        \fancyfoot[R]{\thepage}
        
    }
\fi

\begin{document}

\maketitle

\ifdraftversion
    \thispagestyle{fancy}
\fi

\begin{abstract}
This work formalizes the differential topology of redundancy resolution for systems governed by signed-quadratic actuation maps. By analyzing the minimally redundant case, the global topology of the continuous fiber bundle defining the nonlinear actuation null-space is established. The distribution orthogonal to these fibers is proven to be globally integrable and governed by an exact logarithmic potential field. This field foliates the actuator space, inducing a structural stratification of all orthants into transverse layers whose combinatorial sizes follow a strictly binomial progression. Within these layers, adjacent orthants are continuously connected via lower-dimensional strata termed reciprocal hinges, while the layers themselves are separated by boundary hyperplanes, or portals, that act as global sections of the fibers. This partition formally distinguishes extremal and transitional layers, which exhibit fundamentally distinct fiber topologies and foliation properties. Exploiting this geometric framework, we prove that the orthogonal manifolds within the extremal orthants form a global diffeomorphism to the entire unbounded task space. This establishes the theoretical existence of globally smooth right-inverses that permanently confine the system to a single orthant, guaranteeing the absolute avoidance of kinematic singularities. While motivated by the physical actuation maps of multirotor and marine vehicles, the results provide a strictly foundational topological classification of signed-quadratic surjective systems.
\end{abstract}


\section{Introduction}
\label{sec:introduction}

The redundancy resolution problem requires mapping a lower-dimensional task space onto a higher-dimensional space of actuator internal states~\cite{johansenFossen2013}. In systems governed by signed-quadratic actuation, the physical mapping from the kinetic actuator state to the generated generalized forces is strictly nonlinear and non-conformal. This surjective mapping is the fundamental mathematical model governing the actuation of virtually all propeller-driven platforms, spanning fully-actuated omnidirectional multirotors~\cite{Brescianini2018, odarTMECH2018, hamandiO7plus2020, Lee2025Omnirotor}, morphing and articulated platforms~\cite{omnimorphJINT2024, dragonIJRR2018, yigit2021novel}, modular aerial arrays~\cite{dfaIJRR2013, dodecacopter2025}, and overactuated marine and underwater vehicles~\cite{fossen2011handbook, antonelli2018underwater}.

Traditionally, both the control systems and aerospace communities have approached the inversion of this map as a constrained numerical optimization problem. Static Control Allocation (SCA) instantaneously distributes effort using generalized inverses and algebraic saturation logic~\cite{penrose1955, durham1993, bodson2002, santos2022control}, while Dynamic Control Allocation (DCA) optimizes internal flows strictly within the unobservable null-space of the actuation matrix~\cite{zaccarian2009, serrani2012, cristofaro2017, masocco2023, akbari2025}. Actively manipulating these redundant internal degrees of freedom is also the primary mechanism for fault-tolerant reallocation~\cite{cristofaro2014, liao2025active, baldini2023fault, cai2025} and windup mitigation~\cite{govoni2023}. Because the physical map is nonlinear, these established algorithmic methodologies generally operate by mapping the system into an auxiliary squared state space, artificially reducing the problem to a linear relation. While this pseudo-linear projection successfully satisfies instantaneous algebraic constraints, it systematically ignores the native differential topology of the nonlinear kinetic space. Numerical solvers operate locally, devoid of formal knowledge regarding the global geometry of the constant-task fibers. Consequently, when the auxiliary allocation is mapped back non-conformally to the physical actuator states, trajectories are frequently forced across geometric boundary hyperplanes, inadvertently triggering infinite-derivative kinetic singularities.

To address this foundational gap, this work shifts the paradigm from numerical optimization to differential topology~\cite{boothby2003, isidori1995}. Mathematically, the topology of the signed-quadratic surjection is isomorphic to the direct kinematics of redundant robotic manipulators: both define nonlinear maps from a high-dimensional state space to a lower-dimensional task space, generating infinite continuous fibers of functionally equivalent states. However, unlike classical trigonometric kinematics where the orthogonal distributions are structurally non-integrable (failing the Frobenius condition), the signed-quadratic system possesses a unique, highly structured global geometry that has not been formally classified in the literature.

It is essential to clarify the explicit scope of this work. We do not propose a novel feedback controller, a real-time allocation solver, or a numerical performance benchmark. Instead, we establish the absolute mathematical and differential footprint of the actuator space upon which all future allocation algorithms must theoretically operate. To expose the fundamental atomic unit of this continuous fiber topology, the present analysis is strictly confined to the minimally redundant (codimension-1) case. By isolating this structural baseline, we avoid the combinatorial obfuscation of higher-dimensional stratifications, allowing for exact, analytic proofs of the underlying geometric objects. Algorithms addressing effort optimality, conditioning complexity, and actuator saturation must ultimately be parameterized along the exact boundaries and manifolds derived herein.

The primary contributions of this mathematical framework are as follows:
\begin{enumerate}
    \item \emph{Differential Topology of the Actuation Map:} The geometric structures of the signed-quadratic mapping are formalized. It is proven that the continuous one-dimensional constant-task fibers asymptotically align with and converge to a unique central fiber, establishing the fundamental longitudinal flow of the nonlinear kinetic space.

    \item \emph{Integrability and the Global Logarithmic Potential:} We prove that the task-dimen\-sional distribution orthogonal to the task fibers is globally integrable. The analysis analytically derives the exact global logarithmic potential field whose level sets transversally foliate the actuator space, establishing the existence of exact orthogonal manifolds~\cite{sussmann1973}.
    
    \item \emph{Topological Classification of the Kinetic Space:} A structural stratification of the actuator space into \emph{transitional} and \emph{extremal} orthants is established. The exact boundary topologies of these regions—defining portals, folds, and reciprocal hinges—are characterized, proving that the orthogonal manifolds form a contiguous foliation across singular boundaries. 
    
    \item \emph{Combinatorial Layering and Global Sections:} Transitional orthants are mathematically partitioned into combinatorial groups, termed \emph{layers}, whose sizes are proven to follow a strict binomial progression. Within this architecture, exact global orthogonal sections are analytically constructed.
    
    \item \emph{Existence of Singularity-Free Diffeomorphisms:} It is established that the orthogonal manifolds within the extremal orthants form a global diffeomorphism to the entire unbounded task space. This proves the theoretical existence of a family of globally smooth, static right-inverses that permanently confine the system to a single orthant, guaranteeing the absolute avoidance of kinematic singularities without reliance on heuristic numerical boundary-handling.
\end{enumerate}

The remainder of this paper is organized as follows. Section~\ref{sec:model} formalizes the system model. Section~\ref{sec:fibers_topology} derives the differential topology of the constant-task fibers. Section~\ref{sec:orthogonal_manifolds_orthants} establishes the integrability of the orthogonal distribution and the global logarithmic potential field. Section~\ref{sec:layer_partition} formalizes the combinatorial layer-stratification and constructs the global orthogonal sections. Section~\ref{sec:summary_foliation} summarizes the discovered structures, and Section~\ref{sec:conclusion} concludes the work.

\section{Signed-Quadratic Actuation Model}
\label{sec:model}

This section formalizes the algebraic and differential properties of systems governed by a signed-quadratic actuation map. The mathematical formulation presented herein is strictly agnostic to the underlying physical platform, establishing a general geometric framework for redundancy resolution.

\begin{defn}[Kinetic space and state, Task, and Signed-quadratic Actuation Map]
Let $\mathcal{V} \cong \mathbb{R}^n$ denote the \emph{kinetic space}, representing the array of $n$ internal actuator states $v = [v_1, \dots, v_n]^T$ (the \emph{kinetic state}). Let $\mathcal{W} \cong \mathbb{R}^m$ denote the \emph{task space} of actuator outputs $w = [w_1, \dots, w_m]^T$ (the \emph{task}). A nonlinear signed-quadratic actuation system is characterized by 
an  \emph{actuation map}  defined by:
\begin{equation}
    f: \mathcal{V} \to \mathcal{W},\quad v \mapsto f(v) \triangleq A(v \odot |v|),
    \label{eq:wrench_map}
\end{equation}
where $A \in \mathbb{R}^{m \times n}$ is a constant allocation matrix, $|v| = [|v_1|, \dots, |v_n|]^T$ denotes the element-wise absolute value, and $\odot$ denotes the Hadamard (element-wise) product. 
\end{defn}

To ensure the physical realizability of arbitrary task-space commands and to isolate the native geometric properties of the map, the following structural assumption is established.

\begin{asmpt}[Surjectivity and Minimal  Redundancy]
\label{asmpt:redundancy}
The allocation matrix $A$ is full row rank, such that $\operatorname{rank}(A) = m$. Furthermore, the system is minimally redundant, possessing exactly one degree of actuation redundancy: $n = m + 1$.
\end{asmpt}

The restriction to the minimally redundant case is a deliberate analytical choice. As demonstrated in subsequent sections, this $n = m+1$ configuration uniquely isolates the fundamental one-dimensional fiber topology of the nonlinear null-space, which fundamentally governs the broader redundancy resolution problem.

The continuous execution of a time-varying task trajectory $w(t) \in \mathcal{W}$ requires a corresponding continuous lift $v(t) \in \mathcal{V}$. The relationship between the velocities in the task space and the kinetic space is governed by the differential of the actuation map, $df_v: T_v\mathcal{V} \to T_{f(v)}\mathcal{W}$, which is represented by the state-dependent Jacobian matrix $J(v) = \frac{\partial f}{\partial v} \in \mathbb{R}^{m \times n}$. This tangent-space mapping is given by:
\begin{equation}
    \dot{w} = J(v)\dot{v}.
    \label{eq:differential_map}
\end{equation}
In practical systems, the tangent vector $\dot{v}$ is strictly bounded by finite physical capacities (e.g., $\|\dot{v}\|_\infty \le \bar{u}$). A kinematic singularity occurs at states $v \in \mathcal{V}$ where a chosen global section (right-inverse mapping) from $\mathcal{W}$ to $\mathcal{V}$ fails to be smooth, forcing the required kinetic velocity to diverge ($\|\dot{v}\| \to \infty$) to maintain a bounded task velocity $\dot{w}$. Because physical limits prohibit infinite derivatives, these topological boundaries strictly limit the feasible continuous operational space. Consequently, defining the exact differential topology of $\mathcal{V}$ to guarantee the existence of smooth, singularity-free sections is a strict prerequisite for safe control allocation.

\begin{rem}[Physical Interpretation]
While the map~\eqref{eq:wrench_map} is analyzed geometrically, it inherently models a broad class of propeller-driven architectures. In the context of multirotor UAVs or marine vehicles, the kinetic state $v$ represents the rotor angular velocities, the task $w$ represents the spatial wrenches (forces and moments) exerted on the rigid body, and the matrix $A$ encapsulates the aerodynamic thrust/drag coefficients, rotor positions, and cant angles. The minimal redundancy condition $n=m+1$ corresponds, for example, to a fully actuated heptarotor controlling $m=6$ degrees of freedom, or an underactuated custom planar pentarotor controlling $m=4$ degrees of freedom. Furthermore, the abstract constraint $\|\dot{v}\|_\infty \le \bar{u}$ directly models the physical torque saturation limits of the brushless DC motors. Thus, a geometric singularity demanding $\|\dot{v}\| \to \infty$ physically manifests as actuator windup and likely loss of control.
\end{rem}

\section{Fiber Topology and Tangent Spaces}
\label{sec:fibers_topology}

Given the minimal redundancy condition $n = m + 1$, the level sets of the actuation map $f(v) = w$—that is, the $w$-parameterized sets $f^{-1}(w) = \{v \in \mathbb{R}^n \mid f(v) = w\}$—are one-dimensional continuous curves that foliate the kinetic space $\mathcal{V}$. We formally define the parameterization of these curves.

\begin{defn}[Transformed State and Fiber Parameterization]
\label{def:fiber_param}
Let $g: \mathcal{V}  \to \mathcal{X} \cong \mathbb{R}^n$ be the mapping defined by $v \mapsto g(v) = v \odot |v|$ (i.e., $g_i(v_i) = v_i |v_i|$). This mapping is a global homeomorphism and a diffeomorphism everywhere except on the coordinate hyperplanes (i.e., it fails to be a diffeomorphism wherever $\prod_{i=1}^n v_i = 0$ due to the Jacobian becoming singular). The variable $x=g(v)$ is defined as the \emph{transformed state} of $v$ and $\mathcal{X}$ is the \emph{transformed space}, under which the nonlinear actuation map simplifies to the linear map $f \circ g^{-1} : \mathcal{X} \to \mathcal{W}$ defined by $ x \mapsto Ax$.

For any task $w \in \mathcal{W}$, the level set in the transformed space $\mathcal{X}$ is the one-dimensional affine subspace $\{z_w^A + \lambda b^A \in \mathcal{X} \mid \lambda \in \mathbb{R}\}$, where $z_w^A := A^+ w$ is the minimum-norm particular solution, $A^+$ is the Moore-Penrose pseudoinverse, and $b^A \in \ker(A)\subset \mathcal{X}$ is a unit-norm vector spanning the one-dimensional null-space. 

Applying the inverse mapping $g^{-1}(x) = \operatorname{sign}(x) \odot \sqrt{|x|}$, the explicit parameterization of the \emph{fiber} in the native kinetic space is given by the continuous map $\gamma: \mathcal{W} \times \mathbb{R} \to \mathcal{V}$:
\begin{equation}
    \gamma(w, \lambda) = \operatorname{sign}(z_w^A + \lambda b^A) \odot \sqrt{|z_w^A + \lambda b^A|}.
    \label{eq:fiber}
\end{equation}
\end{defn}

By construction, it is straightforward to verify that $f(\gamma(w, \lambda)) = w$ for all $\lambda \in \mathbb{R}$, confirming that the parameterized curve correctly maps to the constant task, i.e., is the fiber associated to $w$. Varying the scalar parameter $\lambda$, $\gamma(w,\cdot)$ traverses the single dimension of each fiber $f^{-1}(w)$,  $\forall w\in\mathcal{W}$. 

To ensure the physical and mathematical well-posedness of the differential topology across the entire kinetic space, we introduce the following assumption regarding the null-space of the system.

\begin{asmpt}[Strict Redundancy]
\label{asmpt:true_redundancy}
All components of the null-space generator are strictly non-zero, i.e., $b_i^A \neq 0$ for all $i \in \{1, \dots, n\}$.
\end{asmpt}

\begin{rem}
Assumption~\ref{asmpt:true_redundancy} is standard and reflects a properly designed, fault-tolerant redundant actuation system. If $b_k^A = 0$ for some actuator $k$, the null-space condition $A b^A = 0$ dictates that the remaining $n-1 = m$ columns of $A$ are linearly dependent. Consequently, the allocation sub-matrix excluding actuator $k$ would drop to rank $m-1$. Physically, this would render the actuator $k$ a critical single point of failure, as the remaining actuators would lack the actuation span required to maintain full dynamic control of the system. Thus, Assumption~\ref{asmpt:true_redundancy} ensures the system possesses true, fully coupled redundancy.
\end{rem}

\subsection{The Central Fiber and Asymptotic Convergence}
\label{subsec:central_fiber}

We now characterize the global behavior of the fiber foliation within the interior of the kinetic space orthants. For this, we study the specific parameterization of the fibers $\gamma(w, \lambda)$ as the parameter $\lambda$ diverges ($|\lambda| \to \infty$). 

\begin{defn}[Central Fiber and Extremal Orthants]
\label{def:central_fiber}
The \emph{central fiber} is defined as the unique one-dimensional set corresponding to the zero-task condition, $w=0$. Evaluating the inverse mapping $\gamma(0, \lambda)$ yields:
\begin{equation}
    f^{-1}(0) = \left\{ \operatorname{sign}(\lambda b^A) \odot \sqrt{|\lambda b^A|} \in \mathcal{V} \;\middle|\; \lambda \in \mathbb{R} \right\}.
\end{equation}
By defining the characteristic constant vector $c^A := \operatorname{sign}(b^A) \odot \sqrt{|b^A|} \in \mathcal{V}$, the central fiber simplifies to the set $\{ \operatorname{sign}(\lambda)\sqrt{|\lambda|} c^A \in \mathcal{V} \mid \lambda \in \mathbb{R} \}$. Thus, the central fiber is a perfectly straight line in $\mathcal{V}$ passing through the origin and spanned entirely by the direction $c^A$.

Because the central fiber is a straight line, it passes through exactly two opposing orthants in the $n$-dimensional $\mathcal{V}$, defined as the \emph{extremal orthants}. Their geometry is strictly governed by the signature of the null-space generator $b^A$:
\begin{itemize}
    \item The \emph{positive extremal orthant} ($\mathcal{O}^+$) is defined by the sign sequence $\operatorname{sign}(v_i) = \operatorname{sign}(b_i^A)$ for all $i$. The central fiber inhabits this orthant for all $\lambda > 0$.
    \item The \emph{negative extremal orthant} ($\mathcal{O}^-$) is defined by the sign sequence $\operatorname{sign}(v_i) = -\operatorname{sign}(b_i^A)$ for all $i$. The central fiber inhabits this orthant for all $\lambda < 0$.
\end{itemize}
\end{defn}

\begin{thm}[Asymptotic Convergence of Fibers]
\label{thm:asymptotic_convergence}
For any constant task $w \in \mathcal{W}$, the Euclidean distance between the corresponding fiber $f^{-1}(w)$ and the central fiber $f^{-1}(0)$ vanishes asymptotically as the null-space parameter diverges. Specifically:
\begin{equation}
    \lim_{|\lambda| \to \infty} \left\| \gamma(w, \lambda) - \gamma(0, \lambda) \right\| = 0.
\end{equation}
\end{thm}
\begin{proof}
We evaluate the spatial distance between a point on a generic fiber $\gamma(w, \lambda)$ and the corresponding point on the central fiber $\gamma(0, \lambda)$ by examining their difference component-wise. 

For a sufficiently large $|\lambda|$, the null-space component dominates the argument $z_{w,i}^A + \lambda b_i^A$. This has two consequences: first, $\operatorname{sign}(z_{w,i}^A + \lambda b_i^A) = \operatorname{sign}(\lambda b_i^A)$; second, the ratio $\frac{z_{w,i}^A}{\lambda b_i^A}$ approaches zero, guaranteeing that $1 + \frac{z_{w,i}^A}{\lambda b_i^A} > 0$. Because this term is strictly positive, the absolute value can be dropped when factoring out $|\lambda b_i^A|$ from the argument of~\eqref{eq:fiber}. Applying the first-order Taylor expansion $\sqrt{1 + \epsilon} = 1 + \frac{1}{2}\epsilon + \mathcal{O}(\epsilon^2)$ yields:
\begin{align}
\begin{aligned}
\gamma_i(w, \lambda) &= \operatorname{sign}(\lambda b_i^A) \sqrt{|\lambda b_i^A|} \sqrt{1 + \frac{z_{w,i}^A}{\lambda b_i^A}} \\
    &= \gamma_i(0, \lambda) \left( 1 + \frac{z_{w,i}^A}{2 \lambda b_i^A} + \mathcal{O}\left(|\lambda|^{-2}\right) \right).
\end{aligned}
\end{align}
Distributing $\gamma_i(0, \lambda)$ and noting that $\frac{\gamma_i(0, \lambda)}{\lambda b_i^A} = \frac{1}{\sqrt{|\lambda b_i^A|}}$, the spatial difference simplifies to:
$$
\gamma_i(w, \lambda) - \gamma_i(0, \lambda) = \frac{z_{w,i}^A}{2 \sqrt{|\lambda| |b_i^A|}} + \mathcal{O}\left(|\lambda|^{-3/2}\right).
$$
Because the dominant residual term strictly decays proportionally to $|\lambda|^{-1/2}$, the component-wise difference vanishes as $|\lambda| \to \infty$. Thus, the Euclidean distance between the coordinate vectors universally approaches zero.
\end{proof}

Consequently, the central fiber $f^{-1}(0)$ acts as a global asymptote for the entire foliation. Specifically, as $\lambda \to \infty$, the positive half-fibers of all sets $f^{-1}(w)$ converge to the positive half of the central fiber within $\mathcal{O}^+$, and as $\lambda \to -\infty$, the negative half-fibers converge to the negative half of the central fiber within $\mathcal{O}^-$.

\subsection{Tangent Spaces of the Fibers}
\label{subsec:tangent_spaces}

To study the global geometry of these fibers, it is instrumental to analyze the differential mapping.

\begin{lem}[Fiber Tangent Space]
\label{lem:tangent_space}
The Jacobian matrix of the actuation map $f$ evaluated at any kinetic state $v \in \mathcal{V}$ is:
\begin{equation}
    J_f(v) = \frac{\partial f}{\partial v} = 2A \operatorname{diag}(|v_1|, \dots, |v_n|) = 2A D(v),
    \label{eq:jacobian}
\end{equation}
where $D(v) := \operatorname{diag}(|v_1|, \dots, |v_n|)$ acts as a local, state-dependent metric scaling matrix. 

At any regular kinetic state $v$ where $v_i \neq 0$ for all $i \in \{1, \dots, n\}$, the implicit function theorem holds, and the tangent space to the fiber passing through $v$ is a one-dimensional line explicitly spanned by the scaled null-space vector:
\begin{equation}
    T_v f^{-1}(w) = \ker(J_f(v)) = \operatorname{span}\{ D(v)^{-1} b^A \}.
    \label{eq:tangent_regular}
\end{equation}
\end{lem}
\begin{proof}
By the chain rule, the Jacobian is given by $J_f(v) = A \frac{\partial g}{\partial v} = 2A D(v)$. Because $v_i \neq 0$ for all $i$, the diagonal matrix $D(v)$ is strictly positive definite and invertible. Any vector $\tau \in \ker(J_f(v))$ must satisfy $2A D(v) \tau = 0$, which implies $D(v) \tau \in \ker(A) = \operatorname{span}\{b^A\}$. Consequently, it follows that $\tau \in \operatorname{span}\{D(v)^{-1} b^A\}$.
\end{proof}

Having established the algebraic formulation of the tangent spaces in regular domains in Lemma~\ref{lem:tangent_space}, we now analyze their asymptotic behavior and their structural degeneracies at the boundaries of the kinetic orthants.

\begin{prop}[Asymptotic Tangent Alignment]
\label{prop:asymptotic_tangent}
For any constant task $w \in \mathcal{W}$, the tangent spaces of the fibers asymptotically align with the direction of the central fiber as the null-space parameter diverges. Specifically:
\begin{equation}
    \lim_{|\lambda| \to \infty} T_{\gamma(w, \lambda)} f^{-1}(w) = \operatorname{span}\{ c^A \}.
\end{equation}
\end{prop}
\begin{proof}
To evaluate the asymptotic direction of the tangent spaces given by~\eqref{eq:tangent_regular}, we track the components of a generic tangent vector $t(\lambda)$ along the fiber $\gamma(w, \lambda)$. Substituting the explicit fiber definition $\gamma_i(w, \lambda)$ and factoring out the dominant null-space term $\lambda b_i^A$ from the argument yields:
$$
    t_i(\lambda) = \frac{b_i^A}{|\gamma_i(w, \lambda)|} = \frac{b_i^A}{\sqrt{|z_{w,i}^A + \lambda b_i^A|}} = \frac{b_i^A}{\sqrt{|\lambda b_i^A| \left| 1 + \frac{z_{w,i}^A}{\lambda b_i^A} \right|}}.
$$
Recognizing that $\frac{b_i^A}{\sqrt{|b_i^A|}} = \operatorname{sign}(b_i^A)\sqrt{|b_i^A|} = c_i^A$, the tangent component can be rewritten as:
$$
    t_i(\lambda) = \frac{1}{\sqrt{|\lambda|}} c_i^A \left( 1 + \frac{z_{w,i}^A}{\lambda b_i^A} \right)^{-1/2}.
$$
As $|\lambda| \to \infty$, the ratio $\frac{z_{w,i}^A}{\lambda b_i^A}$ approaches zero. By the continuity of the algebraic functions, the residual multiplier strictly converges to $1$. Consequently, as the scalar factor $1/\sqrt{|\lambda|}$ vanishes uniformly across all components $t_i(\lambda)$ for $i=1,\ldots,n$, the exact orientation of the tangent vector converges strictly to $c^A$.
\end{proof}

\begin{thm}[Singularity Alignment at Orthant Boundaries]
\label{thm:singularity_alignment}
Let $I_0 \subset \{1, \dots, n\}$ denote a proper subset of indices ($1 \le |I_0| < n$) representing actuators that have halted, such that $v_i = 0$ for $i \in I_0$. As a fiber approaches this orthant boundary, its tangent space asymptotically aligns entirely with the \emph{degenerate subspace} spanned by the halted axes. Furthermore, for a generic single-actuator crossing ($|I_0| = 1$), the fiber intersects the corresponding boundary hyperplane orthogonally.
\end{thm}
\begin{proof}
At the orthant boundaries, $D(v)$ possesses $|I_0|$ zeros on its main diagonal. Consequently, $J_f(v)$ contains exactly $|I_0|$ zero columns, dropping rank and rendering the standard implicit tangent formulation in~\eqref{eq:tangent_regular} ill-posed. 

To resolve this, we differentiate the explicit geometric parameterization of the fiber $\gamma(w, \lambda)$ with respect to $\lambda$, see~\eqref{eq:fiber}. Applying the chain rule yields the exact components of the parametric tangent vector:
$$
    \frac{\partial \gamma_i}{\partial \lambda} = \frac{b_i^A}{2 \sqrt{|z_{w,i}^A + \lambda b_i^A|}}.
$$
Suppose the actuators in $I_0$ halt simultaneously at some parameter $\lambda^*$, such that $z_{w,i}^A + \lambda^* b_i^A = 0$ for all $i \in I_0$. Near this boundary, the argument of the square root is $z_{w,i}^A + \lambda b_i^A = b_i^A(\lambda - \lambda^*)$. Substituting this into the exact derivative and factoring yields:
\begin{equation}
    \frac{\partial \gamma_i}{\partial \lambda} = \frac{b_i^A}{2 \sqrt{|b_i^A|} \sqrt{|\lambda - \lambda^*|}} = \left( \frac{1}{2 \sqrt{|\lambda - \lambda^*|}} \right) c_i^A.
    \label{eq:gamma_partial_lambda}
\end{equation}
As $\lambda \to \lambda^*$, the scalar multiplier diverges to infinity. Consequently, the tangent components associated with the halted actuators ($i \in I_0$) diverge proportionally to $c_i^A$, while the components of the actively spinning actuators ($i \notin I_0$) remain strictly finite. Because the diverging components infinitely dominate the finite components, the continuous one-dimensional fiber aligns entirely with the subspace spanned by the canonical coordinate axes of $I_0$. If $|I_0| = 1$, this subspace is a single coordinate axis, proving orthogonal intersection with the corresponding boundary hyperplane.
\end{proof}

\begin{cor}[Kinematic Singularity at the Origin]
\label{cor:origin_singularity}
At the zero kinetic state $v = 0 \in \mathcal{V}$, the system encounters a complete \emph{kinematic singularity} where $J_f(0) = 0$. However, by geometric continuation of the central fiber $f^{-1}(0)$, which is a straight line parameterized by $\mu = \operatorname{sign}(\lambda)\sqrt{|\lambda|}$ along $c^A$, the geometric tangent at the origin remains perfectly continuous and uniquely defined as:
\begin{equation}
    T_0 f^{-1}(0) = \operatorname{span}\{ c^A \}.
\end{equation}
\end{cor}

\subsection{Orthant Traversal and Global Topology of the Fibers}
\label{subsec:global_topology}

The explicit parameterization of the fibers reveals a strictly ordered traversal through the kinetic space $\mathcal{V}$. This trajectory dictates how the fibers navigate the orthants and coordinate hyperplanes.

\begin{lem}[Orthant Traversal Sequence]
\label{lem:orthant_traversal}
For any generic task $w \in \mathcal{W}$, the corresponding fiber $\gamma(w, \cdot)$ is strictly monotonic component-wise. It crosses exactly $n$ coordinate hyperplanes and traverses exactly $n+1$ distinct orthants without ever intersecting the same coordinate hyperplane twice.
\end{lem}
\begin{proof}
Consider the internal argument of the fiber parameterization for each actuator, $u_i(w,\lambda) = z_{w,i}^A + \lambda b_i^A$. This argument is a strictly linear function of the null-space parameter $\lambda$. Because $b_i^A \neq 0$ for all $i$ (Assumption~\ref{asmpt:true_redundancy}), each kinetic state component $v_i$ crosses zero exactly once at the critical parameter $\lambda_i^* := -z_{w,i}^A / b_i^A$. 

Each zero crossing represents an intersection with the coordinate hyperplane $v_i=0$, which separates the half-space where $\operatorname{sign}(v_i)=-\operatorname{sign}(b_i^A)$ (traversed for $\lambda<\lambda_i^*$) from the half-space where $\operatorname{sign}(v_i)=\operatorname{sign}(b_i^A)$ (traversed for $\lambda>\lambda_i^*$). Sorting these $n$ distinct roots on $\mathbb{R}$ identifies $n$ unique hyperplane crossings. These $n$ crossings divide the parametric domain $\lambda \in \mathbb{R}$ into exactly $n+1$ continuous intervals, corresponding to $n+1$ traversed orthants out of the $2^n$ possible orthants.
\end{proof}

\begin{rem}[Non-Generic Tasks and Simultaneous Crossings]
\label{rem:non_generic_tasks}
Lemma~\ref{lem:orthant_traversal} relies on the generic condition that all $n$ roots $\lambda_i^*$ are strictly distinct. The set of non-generic tasks—where two or more roots coincide—forms a subset of measure zero in $\mathcal{W}$. If $m > 1$ roots coincide at a single parameter value $\lambda^*$, the fiber intersects an $(n-m)$-dimensional stratum (the simultaneous intersection of $m$ coordinate hyperplanes). During this measure-zero event, $m$ coordinates change sign instantaneously, causing the fiber to ``skip'' $m-1$ intermediate orthants. Consequently, the fiber traverses $k+1$ orthants, where $k$ (with $k < n$) is the number of strictly distinct roots.
\end{rem}

\begin{defn}[Transitional Orthants]
\label{def:transitional_orthants}
Based on Lemma~\ref{lem:orthant_traversal}, the initial orthant traversed for $\lambda < \min_i\{\lambda_i^*\}$ corresponds to the negative extremal orthant $\mathcal{O}^-$. The final orthant traversed for $\lambda > \max_i\{\lambda_i^*\}$ corresponds to the positive extremal orthant $\mathcal{O}^+$. For a generic task, the $n-1$ intermediate orthants traversed between $\mathcal{O}^-$ and $\mathcal{O}^+$ are defined as the \emph{transitional orthants}. There exist $2^n-2$ such transitional orthants in the topology of the map, and each generic fiber with $w\neq0$ traverses exactly $n-1$ of them.
\end{defn}

\begin{prop}[Strict Forward Progression]
\label{prop:forward_progression}
The unidirectional traversal of any fiber is geometrically characterized by a continuous forward progression relative to the central fiber. The projection of the fiber's tangent vector onto the central fiber direction $c^A$ is strictly positive everywhere, ensuring the fiber never stalls or reverses direction.
\end{prop}
\begin{proof}
We evaluate the inner product of the parametric tangent vector $\frac{\partial \gamma}{\partial \lambda}$ and the globally constant direction of the central fiber $c^A$:
$$
\left\langle \frac{\partial \gamma}{\partial \lambda}, c^A \right\rangle = \sum_{i=1}^n \left( \frac{b_i^A}{2|v_i|} \right) \left( \frac{b_i^A}{\sqrt{|b_i^A|}} \right) = \sum_{i=1}^n \frac{(b_i^A)^2}{2|v_i|\sqrt{|b_i^A|}}.
$$
Because $(b_i^A)^2 > 0$ and the absolute value terms are strictly positive for $v_i \neq 0$, the summation is strictly positive everywhere within the interior of the orthants. Furthermore, at the intersection with any coordinate hyperplane (as $\lambda \to \lambda_i^*$), the corresponding velocity component $v_i \to 0$, causing its associated term in the summation to diverge to $+\infty$. This guarantees that the forward progression is strictly maintained at the boundaries. 
\end{proof}

Geometrically, Proposition~\ref{prop:forward_progression} establishes that the angle between the central fiber and the tangent space of each generic fiber is strictly acute. While the generic fiber curves through the kinetic space, its progression along the longitudinal subspace spanned by $c^A$ is monotonically increasing.

This bounded traversal fundamentally defines the mapping topology between the kinetic space $\mathcal{V}$ and the task space $\mathcal{W}$, summarized as follows.

\begin{thm}[Mapping Surjectivity of the Orthants]
\label{thm:mapping_surjectivity}
The restriction of the actuation mapping $f(v) = A(v \odot |v|)$ to either extremal orthant ($\mathcal{O}^+$ or $\mathcal{O}^-$) is globally surjective onto $\mathcal{W}$. Conversely, the restriction of $f$ to any transitional orthant maps onto a limited conical sector of $\mathcal{W}$ that strictly excludes a neighborhood of the origin.
\end{thm}
\begin{proof}
As $\lambda \to \pm\infty$, the signs of the state components are asymptotically governed by $\pm \operatorname{sign}(b_i^A)$. Consequently, every continuous fiber universally originates in $\mathcal{O}^-$ and terminates in $\mathcal{O}^+$. Because every generic task $w \in \mathcal{W}$ generates a fiber that asymptotically occupies these two sets, both extremal orthants contain a pre-image for every $w \in \mathcal{W}$, proving global surjectivity.

In contrast, Lemma~\ref{lem:orthant_traversal} dictates that the transitional orthants are traversed only by a proper subset of the fibers. Geometrically, any orthant is a conical set in $\mathcal{V}$ (closed under positive scaling $v \mapsto kv$ for $k>0$). Because the map is positively homogeneous of degree 2, meaning $f(k v) = k^2 f(v)$, it preserves this scale invariance; if $w$ is in the image of an orthant, so is any positive scaling $\alpha w$ for $\alpha > 0$, proving the image is a conical sector in $\mathcal{W}$. Furthermore, the exact pre-image of the zero-task state ($w=0$) is the central fiber. By Definition~\ref{def:central_fiber}, the central fiber resides fully within the extremal orthants and intersects the transitional orthants only at the origin $v=0$. Thus, the image of the interior of any transitional orthant inherently excludes $w=0$ and, by continuity, a surrounding neighborhood in $\mathcal{W}$.
\end{proof}
 
 Figure~\ref{fig:task_fibers} illustrates the 1D continuous fiber bundles for four distinct allocation configurations in $n=2$ and $n=3$ kinetic spaces, encompassing symmetric, asymmetric, and skewed geometries. This visualizes the global topological structure of the map, specifically highlighting the characteristic expansion near singular boundaries and asymptotic alignment with the central fiber.

\begin{figure*}[t]
     \centering
     \begin{subfigure}[b]{0.24\linewidth}
         \centering
         \includegraphics[width=\linewidth]{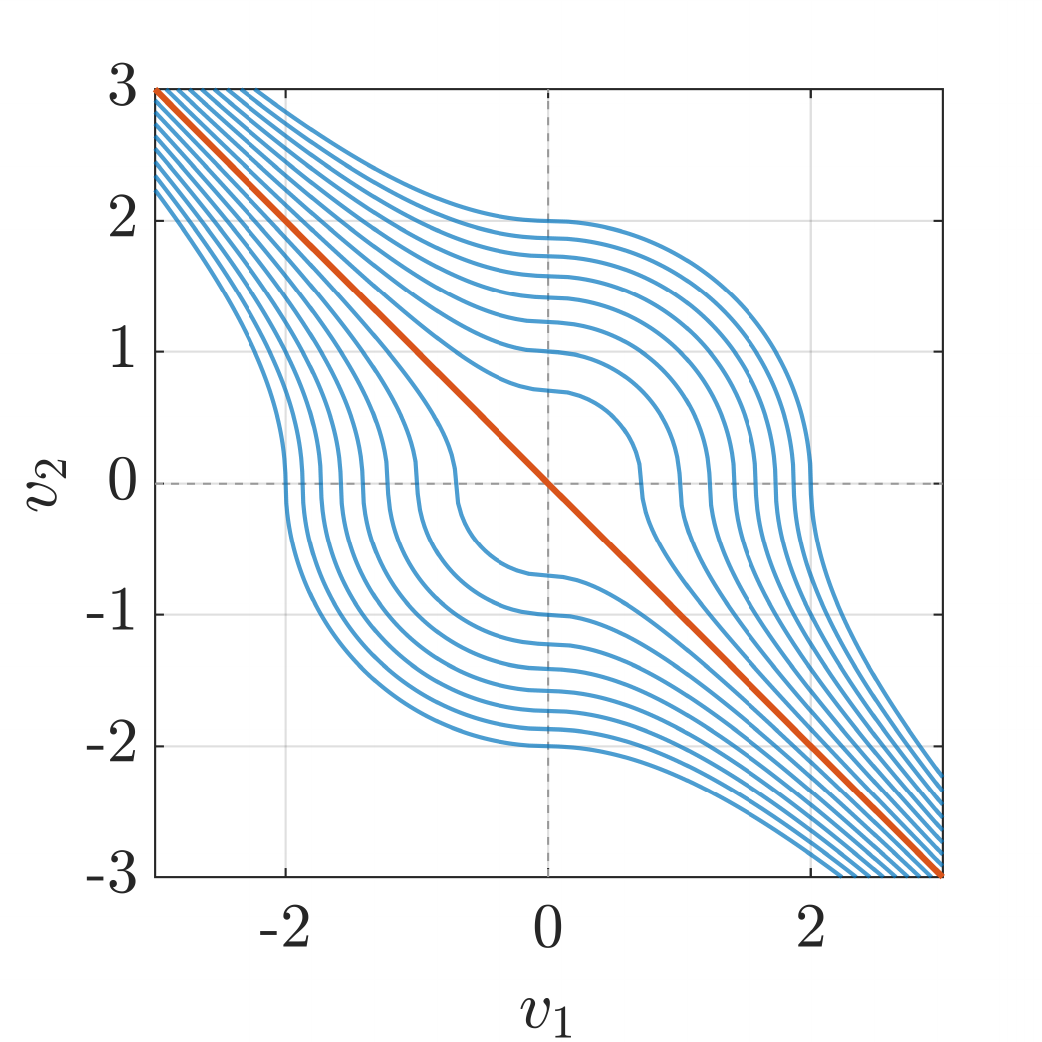}
         \caption{$n=2$ (Sym.)}
         \label{subfig:fibers_n2_sym}
     \end{subfigure}
     \hfill
     \begin{subfigure}[b]{0.24\linewidth}
         \centering
         \includegraphics[width=\linewidth]{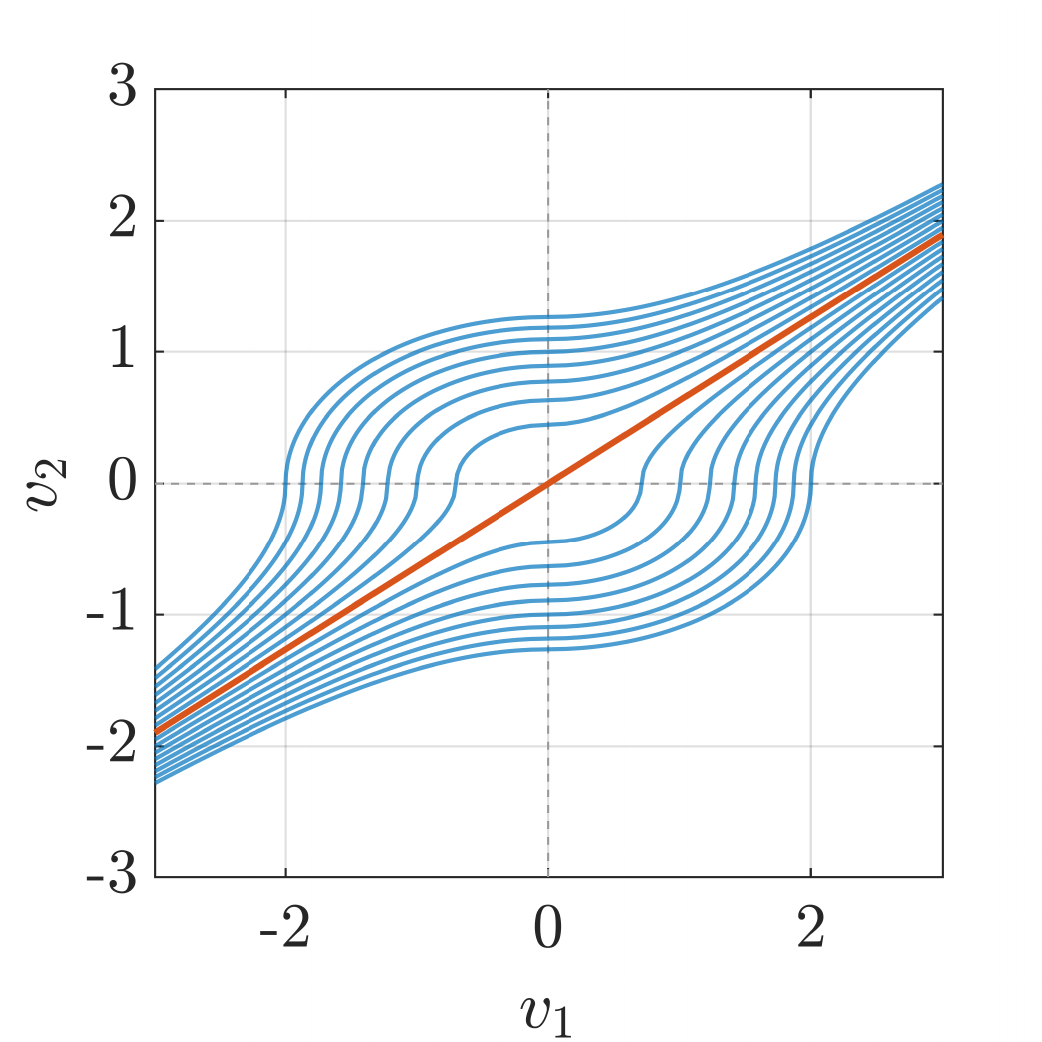}
         \caption{$n=2$ (Asym.)}
         \label{subfig:fibers_n2_asym}
     \end{subfigure}
     \begin{subfigure}[b]{0.24\linewidth}
         \centering
         \includegraphics[width=\linewidth]{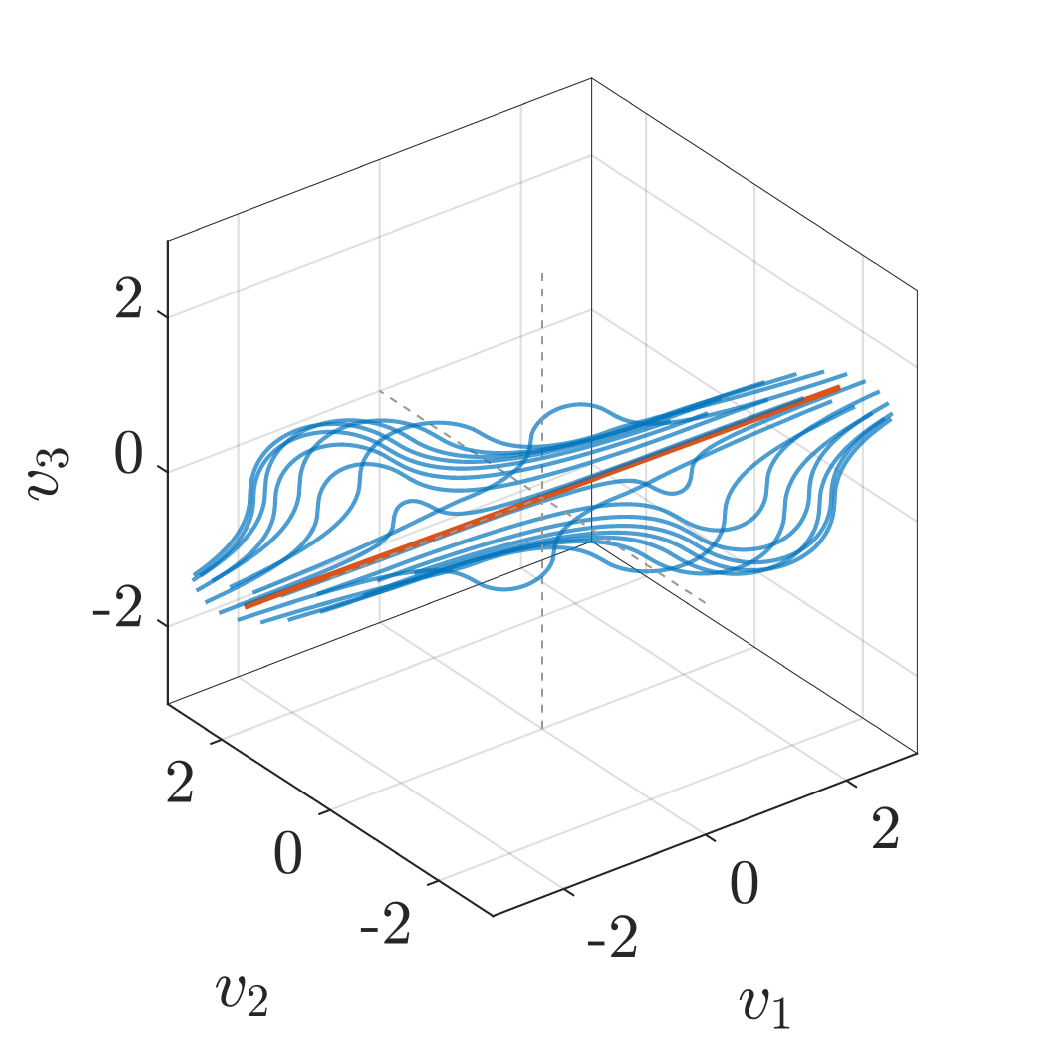}
         \caption{$n=3$ (Std.)}
         \label{subfig:fibers_n3_A3}
     \end{subfigure}
     \hfill
     \begin{subfigure}[b]{0.24\linewidth}
         \centering
         \includegraphics[width=\linewidth]{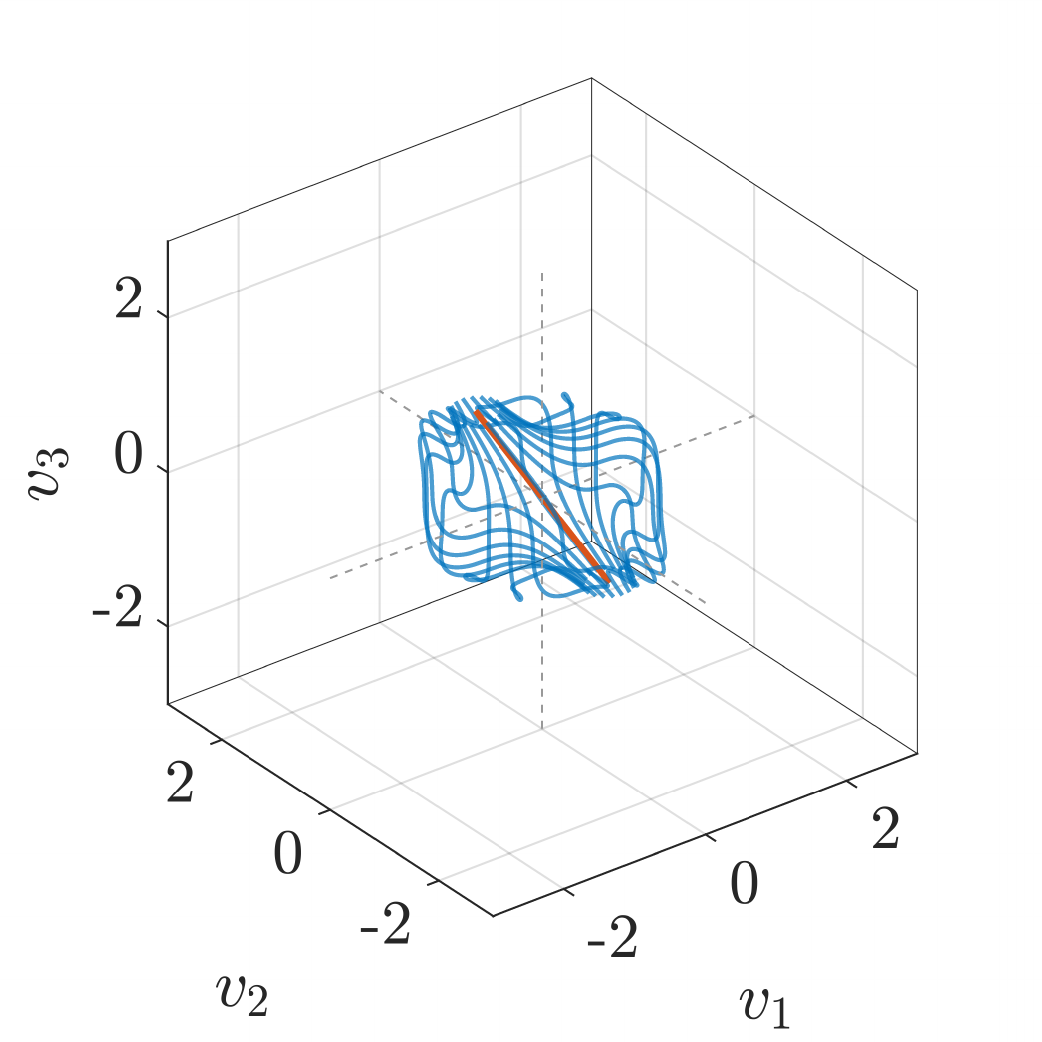}
         \caption{$n=3$ (Skew)}
         \label{subfig:fibers_n3_A4}
     \end{subfigure}
    
     \caption{The continuous fiber bundle defining the actuation null-space in the native kinetic space $\mathcal{V}$. The central fiber (red, thick) anchors the longitudinal flow. Adjacent fibers exhibit the characteristic expansion near the origin before asymptotically converging to and aligning with the central fiber.}
     \label{fig:task_fibers}
 \end{figure*}

\section{Orthogonal Manifolds in the Orthants}
\label{sec:orthogonal_manifolds_orthants}

Having established the topological structure of the constant-task fibers $f^{-1}(w): \lambda \mapsto \gamma(w,\lambda)$, we now analyze and discover the existence and topology of $m$-dimensional manifolds orthogonal to these fibers.
Navigating these orthogonal manifolds represents the purely task-actuating motions of the system in $\mathcal{W}$, while moving strictly along the fibers represents internal `null-space' reconfigurations in $\mathcal{V}$, which are `invisible' in $\mathcal{W}$.

\subsection{Integrability and the Global Logarithmic Potential}

At any regular kinetic state $v \in \mathcal{V}$ (where $v_i \neq 0$ for all $i$), the tangent space to the one-dimensional fiber is given by~\eqref{eq:tangent_regular}. We seek to define an $m$-dimensional manifold $\mathcal{M}$ passing through $v$ that is strictly orthogonal to the fibers at every point. The tangent space of $\mathcal{M}$ at $v$ must satisfy:
\begin{equation}
    T_v\mathcal{M} = \ker(J_f(v))^\perp = \operatorname{Im}(J_f^T(v)) = \operatorname{span}\{ D(v) A^T \}.
    \label{eq:tangent_space_orthogonal_manifold}
\end{equation}
In differential geometric terms, finding this manifold equates to integrating the distribution of the $m$-dimensional vector spaces $T_v\mathcal{M}$ orthogonal to the fibers.

\begin{rem}
\label{rem:metric}
Exploiting the presence of the standard metric in $\mathbb{R}^n$, we algebraically treat tangent spaces and cotangent spaces as residing in $\mathbb{R}^n$, with their dual product operating as the standard inner product.
\end{rem}

\begin{thm}[Integrability of the Orthogonal Distribution]
\label{thm:integrability}
The distribution of vector spaces orthogonal to the fibers, $\Delta(v) = \ker(J_f(v))^\perp$, is globally integrable on $\mathbb{R}^n \setminus \{v_i = 0\}$. The resulting integral manifolds are uniquely defined by the level sets of a global, exact scalar potential field $\Phi: \mathbb{R}^n \setminus \{v_i = 0\} \to \mathbb{R}$, given by:
\begin{equation}
    \Phi(v) = \sum_{i=1}^n b_i^A \operatorname{sign}(v_i) \ln|v_i|.
    \label{eq:log_potential}
\end{equation}
\end{thm}
\begin{proof}
For a differential state displacement $dv \in \mathbb{R}^n$ to be strictly orthogonal to the fiber at $v$, its inner product with the fiber's tangent vector must vanish:
\[
    \left\langle dv, D(v)^{-1} b^A \right\rangle = 0 \implies \sum_{i=1}^n b_i^A \frac{dv_i}{|v_i|} = 0.
\]
Because the $i$-th term of this sum depends exclusively on the $i$-th coordinate $v_i$, the differential form is exact. Integrating the differential equation component-wise yields the exact, global potential function $\Phi(v)$ in~\eqref{eq:log_potential}, proving that the distribution forms well-defined, continuous global manifolds rather than a non-holonomic vector field.
\end{proof}

\begin{defn}[Orthogonal Manifolds]
\label{def:orthogonal_manifolds}
The \emph{orthogonal manifolds} are defined as the level sets $\mathcal{M}_C = \{v \in \mathbb{R}^n \setminus \{v_i = 0\} \mid \Phi(v) = C\}$ for any constant $C \in \mathbb{R}$. Because the gradient $\nabla \Phi(v) = D(v)^{-1} b^A$ never vanishes or explodes for any finite $v$ in the interior of an orthant (as $b_i^A \neq 0$), the level sets $\mathcal{M}_C$ form a non-intersecting foliation of the open kinetic space, where each $\mathcal{M}_C$ is an $m$-dimensional smooth manifold.
\end{defn}

\subsection{Transverse Foliation and Boundary Limits}

We now establish the algebraic relationship between the orthogonal manifold level $C$ and the internal `null-space' parameter $\lambda$ of a given fiber. Substituting the explicit fiber parameterization $\gamma(w, \lambda)$ from~\eqref{eq:fiber} into the global potential function~\eqref{eq:log_potential} yields the constant $C$ of the exact orthogonal manifold pierced by the fiber at a given $\lambda$:
\begin{equation}
\begin{aligned}
    &C_w(\lambda) := \Phi(\gamma(w, \lambda)) = \\
    &\frac{1}{2} \sum_{i=1}^n b_i^A \operatorname{sign}\big((z_w^A)_i + \lambda b_i^A\big) \ln\big|(z_w^A)_i + \lambda b_i^A\big|.
\end{aligned}
\label{eq:C_lambda_relation}
\end{equation}

\begin{prop}[Transverse Foliation and Monotonicity]
\label{prop:monotonicity}
Within the interior of any orthant of $\mathcal{V}$, the scalar potential evaluated along a fiber, $C_w(\lambda)$, is strictly monotonically increasing with respect to $\lambda$. Consequently, each orthogonal manifold crosses every continuous fiber exactly once within a given orthant.
\end{prop}
\begin{proof}
The derivative of $C_w$ with respect to $\lambda$ is computed via the chain rule, $\frac{dC_w}{d\lambda} = \sum_{i=1}^n \frac{\partial \Phi}{\partial v_i} \frac{dv_i}{d\lambda}$. Given $\frac{\partial \Phi}{\partial v_i} = \frac{b_i^A}{|v_i|}$ and the local tangent component $\frac{dv_i}{d\lambda} = \frac{b_i^A}{2|v_i|}$ derived from~\eqref{eq:fiber}, we obtain:
\[ 
\frac{dC_w}{d\lambda} = \sum_{i=1}^n \frac{(b_i^A)^2}{2 v_i^2} > 0.
\]
Because $b_i^A \neq 0$ and $v_i \neq 0$ in the interior, this derivative is strictly positive, guaranteeing strict monotonicity.
\end{proof}

While Proposition~\ref{prop:monotonicity} establishes local monotonicity, the global topology of the foliation is strictly governed by the boundary conditions at the coordinate hyperplanes.

\begin{thm}[Potential Limits at Orthant Boundaries]
\label{thm:potential_limits}
Let a fiber $\gamma(w,\cdot)$ traverse a generic transitional orthant, entering at parameter $\lambda_{in}$ across the hyperplane $v_j = 0$, and exiting at $\lambda_{out} > \lambda_{in}$ across the hyperplane $v_k = 0$. The potential function spans the entire real line within this single orthant: $\lim_{\lambda \to \lambda_{in}^+} C_w(\lambda) = -\infty$ and $\lim_{\lambda \to \lambda_{out}^-} C_w(\lambda) = +\infty$.
\end{thm}
\begin{proof}
At the entry boundary ($\lambda \to \lambda_{in}^+$), the transformed coordinate $x_j(\lambda) = (z_w^A)_j + \lambda b_j^A$ departs from zero. Because the fiber advances monotonically, $x_j$ immediately assumes the sign of its directional slope, $b_j^A$. Thus, just inside the new orthant, $\operatorname{sign}(v_j) = \operatorname{sign}(b_j^A)$. The $j$-th term of the potential becomes $|b_j^A| \ln|v_j|. As |v_j| \to 0$, $\ln|v_j| \to -\infty$, dominating the sum and yielding $-\infty$.

Conversely, as the fiber approaches the exit boundary ($\lambda \to \lambda_{out}^-$), $x_k(\lambda)$ approaches zero from a non-zero value. Because $\lambda$ is strictly increasing, $x_k$ is moving in opposition to its slope $b_k^A$, imposing $\operatorname{sign}(v_k) = -\operatorname{sign}(b_k^A)$. The $k$-th term in the potential becomes $-|b_k^A| \ln|v_k|$. As $|v_k| \to 0$, the logarithmic divergence is inverted by the sign coefficient, forcing the sum to $+\infty$.
\end{proof}

Geometrically, Theorem~\ref{thm:potential_limits} governs the asymptotic behavior of the orthogonal manifolds near the orthant boundaries. While the isolated vanishing of a single kinetic state coordinate forces the manifold to flatten as its potential diverges, it does not inherently compromise the system's control authority. In contrast, any continuous trajectory strictly confined to a specific orthogonal manifold (maintaining a constant, finite potential $C$) precludes such isolated vanishing; the logarithmic divergence must be balanced by at least one additional coordinate approaching zero with an opposing sign coefficient. This mandatory simultaneous vanishing fundamentally restricts the tangent space geometry at the boundary intersections, ultimately inducing a strict kinematic singularity. These exact geometric and algebraic consequences are formalized in the following corollaries.

\begin{cor}[Flattening of the Orthogonal Manifold]
\label{cor:flattening_manifold}
As a single coordinate $v_i$ of the kinetic state $v$ tends to $0$, the orthogonal manifold passing through $v$ asymptotically flattens to become parallel to the bounding hyperplane $v_i = 0$, while its corresponding potential diverges, $|C| \to \infty$. Note that the vanishing of a single coordinate does not induce a singularity; the Jacobian in~\eqref{eq:jacobian} remains full rank by virtue of Assumption~\ref{asmpt:true_redundancy} and thus it is possible to select $\dot{v}$ parallel to the bounding hyperplane while generating all possible velocities in $T_{f(v)}\mathcal{W}$.
\end{cor}

\begin{cor}[Geometric Boundary Alignment]
\label{cor:geometric_alignment}
For a continuous path confined to an orthogonal manifold of finite potential $C$, the approach to a bounding hyperplane $v_i \to 0$ necessitates the simultaneous vanishing of at least one additional coordinate $v_j \to 0$ ($j \neq i$). In the limit as $k \ge 2$ coordinates simultaneously vanish, the corresponding components of the gradient $\nabla \Phi= D(v)^{-1} b^A$ diverge. Consequently, the normal vector of the manifold aligns entirely within the $k$-dimensional subspace of the vanishing coordinates. The $m$-dimensional tangent space of the manifold therefore strictly aligns to contain the $(n-k)$-dimensional intersection of those bounding hyperplanes, with its remaining dimensions extending outward into the interior of the orthant. 
\end{cor}

\begin{cor}[Kinematic Singularity at the Boundary]
\label{cor:kinematic_singularity}
The necessary simultaneous vanishing of $k \ge 2$ coordinates at the boundary of the orthogonal manifolds, as dictated by Corollary~\ref{cor:geometric_alignment}, causes the Jacobian in~\eqref{eq:jacobian} to lose rank. This rank deficiency constitutes a kinematic singularity of the nonlinear mapping, locally precluding the instantaneous assignment of task velocities $\dot{w}$ within a $(k-1)$-dimensional subspace of the tangent space $T_{f(v)}\mathcal{W}$.
\end{cor}

\subsection{Topology of the Orthogonal Foliation for Transitional Orthants}
\label{subsec:global_topology_foliation_transitional}

To formally characterize the global evolution of the orthogonal manifolds $\mathcal{M}_C$ and their interaction with the task fibers $\gamma(w, \lambda)$, we partition the boundary of any transitional orthant into functional sets. Let $\mathcal{O}_\sigma$ be a transitional orthant defined by the sign vector $\sigma \in \{-1, 1\}^n$, with boundary faces defined as $F^k_\sigma = \{v \in \partial \mathcal{O}_\sigma \mid v_k = 0\}$. 

\begin{defn}[Portals, Hinges, and Folds]
\label{def:portals_hinges}
Based on the mapping of the null-space vector $b^A$, the coordinate indices $k \in \{1, \dots, n\}$ within $\mathcal{O}_\sigma$ are partitioned into an entry set $\mathcal{I}^{in}_\sigma = \{i \mid \sigma_i b_i^A > 0\}$ and an exit set $\mathcal{I}^{out}_\sigma = \{j \mid \sigma_j b_j^A < 0\}$. 

We define the \emph{entry portal} $\mathcal{P}^{in}_\sigma$ and the \emph{exit portal} $\mathcal{P}^{out}_\sigma$ as the unions of their respective boundary faces:
\begin{equation}
    \mathcal{P}^{in}_\sigma = \bigcup_{i \in \mathcal{I}^{in}_\sigma} F^i_\sigma, \quad \mathcal{P}^{out}_\sigma = \bigcup_{j \in \mathcal{I}^{out}_\sigma} F^j_\sigma.
\end{equation}
The lower-dimensional strata formed by the intersections of these faces are partitioned into two distinct geometric classes:
\begin{enumerate}
    \item The \emph{orthant hinges} $\mathcal{H}_\sigma$: The intersections between entry and exit faces, defined as $\mathcal{H}_\sigma = \mathcal{P}^{in}_\sigma \cap \mathcal{P}^{out}_\sigma$.
    \item The \emph{portal folds} $\mathcal{S}^{in}_\sigma$ and $\mathcal{S}^{out}_\sigma$: The face intersections contained entirely within a single portal, lacking any hinge component. For example, the entry folds are defined as $\mathcal{S}^{in}_\sigma = \bigcup_{i \neq k \in \mathcal{I}^{in}_\sigma} (F^i_\sigma \cap F^k_\sigma) \setminus \mathcal{H}_\sigma$.
\end{enumerate}
\end{defn}

\begin{thm}[Topological Sweep of the Orthogonal Foliation]
\label{thm:topological_sweep}
For any transitional orthant $\mathcal{O}_\sigma$, the task fibers $\gamma(w, \lambda)$ and the orthogonal manifolds $\mathcal{M}_C$ exhibit the following topological properties:
\begin{enumerate}
    \item Every fiber that traverses the interior of $\mathcal{O}_\sigma$ enters strictly through $\mathcal{P}^{in}_\sigma$ (including its folds $\mathcal{S}^{in}_\sigma$) and exits strictly through $\mathcal{P}^{out}_\sigma$ (including $\mathcal{S}^{out}_\sigma$). Fibers intersecting the hinges $\mathcal{H}_\sigma$ are strictly tangent to $\bar{\mathcal{O}}_\sigma$ and map to the task-space boundary of achievable tasks $\partial \mathcal{W}_\sigma$ for that specific orthant.
    \item For any finite scalar $C \in \mathbb{R}$, the boundary of the orthogonal manifold $\mathcal{M}_C$ is asymptotically anchored entirely and exclusively to the hinges $\mathcal{H}_\sigma$.
    \item As $C \to -\infty$, the manifold $\mathcal{M}_C$ asymptotically flattens against the entry portal $\mathcal{P}^{in}_\sigma$. As $C \to +\infty$, $\mathcal{M}_C$ asymptotically flattens against the exit portal $\mathcal{P}^{out}_\sigma$, continuously sweeping the entire interior of the orthant. During this asymptotic flattening, the manifolds develop geometric creases along the portal folds ($\mathcal{S}^{in}_\sigma$ and $\mathcal{S}^{out}_\sigma$) to attain contiguous adherence to the multi-faced portals.
\end{enumerate}
\end{thm}

\begin{proof}
We prove the three claims sequentially. 

\textit{1)} Let a fiber be parameterized by $\lambda$, mapping to the transformed state $x_k(\lambda) = (z_w^A)_k + \lambda b_k^A$. For a fiber to enter $\mathcal{O}_\sigma$ across a face $F^k_\sigma$, the coordinate must transition to match the orthant sign $\sigma_k$, requiring $\operatorname{sign}(dx_k/d\lambda) = \operatorname{sign}(b_k^A) = \sigma_k$. This precisely satisfies the condition for $\mathcal{I}^{in}_\sigma$, proving entry occurs exclusively through $\mathcal{P}^{in}_\sigma$. Conversely, exiting satisfies the condition for $\mathcal{I}^{out}_\sigma$. 

Now, consider a fiber intersecting a hinge at $h \in \mathcal{H}_\sigma$ at some parameter $\lambda^*$. By definition of $\mathcal{H}_\sigma$, there exists at least one entry index $i \in \mathcal{I}^{in}_\sigma$ and one exit index $j \in \mathcal{I}^{out}_\sigma$ such that $x_i(\lambda^*) = x_j(\lambda^*) = 0$ simultaneously. For any infinitesimal progression $\lambda > \lambda^*$, the entry coordinate $x_i$ transitions to match $\sigma_i$. However, because $j \in \mathcal{I}^{out}_\sigma$, the exit coordinate $x_j$ strictly transitions away from $\sigma_j$. Therefore, for all $\lambda \neq \lambda^*$, the fiber violates the strict sign condition of $\mathcal{O}_\sigma$. The fiber remains entirely outside the orthant interior, intersecting it tangentially without penetrating. These tangent fibers define the geometric boundary $\mathcal{W}_\sigma=\partial f(\mathcal{O}_\sigma)$ in the task-space, separating the tasks $w\in f(\mathcal{O}_\sigma)$ that are achievable by a kinetic state $v\in\mathcal{O}_\sigma$ from those that are not. At folds $\mathcal{S}^{in}_\sigma$, all intersecting faces belong to $\mathcal{I}^{in}_\sigma$; thus, all concurrent zero-crossings transition into the correct orthant signs simultaneously, permitting valid entry. Similarly for $\mathcal{S}^{out}_\sigma$ regarding the exit.

\textit{2)} The orthogonal manifolds inside $\mathcal{O}_\sigma$ are defined by the level sets of the potential function:
\begin{equation}
    \Phi(v) = \sum_{i \in \mathcal{I}^{in}_\sigma} |b_i^A| \ln|v_i| - \sum_{j \in \mathcal{I}^{out}_\sigma} |b_j^A| \ln|v_j| = C.
    \label{eq:potential_in_out}
\end{equation}
Consider the limit as the manifold approaches an entry face $F^i_\sigma$ (where $v_i \to 0$) for a fixed, finite $C$. The corresponding entry logarithmic term diverges to $-\infty$. To satisfy the potential equation $\Phi(v) = C$, the exit summation must correspondingly diverge to $-\infty$, which strictly requires at least one exit coordinate $v_j \to 0$ (approaching face $F^j_\sigma$). Therefore, the limits of the manifold exist only where $v_i = 0$ and $v_j = 0$ simultaneously, which precisely defines the hinges $\mathcal{H}_\sigma = F^i_\sigma \cap F^j_\sigma$.

\textit{3)} Let $C \to -\infty$. To satisfy the potential equation, the strictly positive summation over the entry set in~\eqref{eq:potential_in_out} must dominate, forcing $\sum_{i \in \mathcal{I}^{in}_\sigma} |b_i^A| \ln|v_i| \to -\infty$. This requires the manifold to be arbitrarily close to the entry faces, meaning $\mathcal{M}_{-\infty} \to \mathcal{P}^{in}_\sigma$. Because $\mathcal{P}^{in}_\sigma$ is the union of multiple intersecting hyperplanes, a single continuous manifold flattening against it must conform to its piecewise-linear topology. Consequently, the manifold creases along the intersection loci of these faces, mapping perfectly to the pure entry folds $\mathcal{S}^{in}_\sigma$. By symmetry, as $C \to +\infty$, the subtracted exit sum in~\eqref{eq:potential_in_out}  dominates, forcing $\mathcal{M}_{+\infty} \to \mathcal{P}^{out}_\sigma$, and the manifold creases along $\mathcal{S}^{out}_\sigma$. Because $\Phi(v)$ is continuous and monotonically increasing along any valid transversing fiber (Proposition~\ref{prop:monotonicity}), the manifolds continuously sweep the entire volume of $\mathcal{O}_\sigma$ between these asymptotic bounds.
\end{proof}

\subsection{Topology of the Orthogonal Foliation for Extremal Orthants and Asymptotic Divergence}
\label{subsec:global_topology_foliation_extremal}

The topological behavior defined for transitional orthants shifts fundamentally when considering an extremal orthant. An orthant is defined as extremal if its sign vector $\sigma$ perfectly matches (the positive orthant $\mathcal{O}^+$) or perfectly opposes (the negative orthant $\mathcal{O}^-$) the signs of the null-space mapping vector $b^A$. 

\begin{thm}[Topological Sweep in Extremal Orthants]
\label{thm:topological_sweep_extremal}
In the extremal orthants $\mathcal{O}^+$ and $\mathcal{O}^-$, there are no hinges ($\mathcal{H}^+ = \mathcal{H}^- = \emptyset$). Furthermore, $\mathcal{O}^+$ possesses no exit portals ($\mathcal{P}^{out}_+ = \emptyset$), and $\mathcal{O}^-$ possesses no entry portals ($\mathcal{P}^{in}_- = \emptyset$). For any constant finite $C$, the orthogonal manifolds $\mathcal{M}_C$ do not converge asymptotically to any hinge. Instead, as $C$ varies, the manifolds sweep the entire infinite volume of the orthant: they asymptotically flatten against the bounding portals in one limit of $C$, and expand with unbounded divergence toward spatial infinity in the opposite limit.
\end{thm}

\begin{proof}
In an extremal orthant, all components of the product $\sigma_k b_k^A$ share the same strict sign. Consequently, all boundary faces belong exclusively to the entry portal $\mathcal{P}^{in}_\sigma$ for $\mathcal{O}^+$ (where $\sigma_k b_k^A > 0$) or exclusively to the exit portal $\mathcal{P}^{out}_\sigma$ for $\mathcal{O}^-$ (where $\sigma_k b_k^A < 0$). The opposing portal is inherently empty, which strictly implies the intersection set is empty: $\mathcal{H}_\sigma = \mathcal{P}^{in}_\sigma \cap \mathcal{P}^{out}_\sigma = \emptyset$.

Consider the extremal orthant $\mathcal{O}^+$, where the exit portal is empty ($\mathcal{P}^{out}_+ = \emptyset$). Within this orthant, the potential function evaluates as a strictly positively weighted sum of logarithms: $\Phi(v) = \sum_{i=1}^n |b_i^A| \ln|v_i| = C$.

To formalize the limit $C \to -\infty$, consider an arbitrary $v^*\in \mathcal{P}^{in}_+$. By definition of the boundary, $v^*$ possesses at least one vanishing coordinate, $v^*_k = 0$. For any sequence of interior points $v^{(m)} \to v^*$, the corresponding term $|b_k^A| \ln|v^{(m)}_k| \to -\infty$. Because all other logarithmic terms are bounded from above in any finite neighborhood of $v^*$, the total potential strictly diverges to $-\infty$. Since this holds for every point $v^* \in \mathcal{P}^{in}_+$, the level set $\mathcal{M}_C$ asymptotically converges to and entirely covers the portal as $C \to -\infty$.

Conversely, to formalize the limit $C \to +\infty$, consider an arbitrary open ray strictly within the interior of the orthant, parameterized by $v(r) = r u$, where $u \in \operatorname{int}(\mathcal{O}^+)$ is a constant directional vector and $r > 0$. Evaluating the potential along this ray yields:
\begin{equation}
    \Phi(ru) = \left( \sum_{i=1}^n |b_i^A| \right) \ln r + \sum_{i=1}^n |b_i^A| \ln u_i.
\end{equation}
Because $u$ lies strictly in the interior, $u_i > 0$ for all $i$, making the second sum a finite constant. As the radial distance $r \to \infty$, the strictly positive coefficient $\sum |b_i^A| > 0$ guarantees that $\Phi(ru) \to +\infty$. Because this strict radial divergence applies along \emph{any} arbitrary interior direction $u$, the sequence of manifolds $\mathcal{M}_C$ expands uniformly outward in all directions, encompassing all points at spatial infinity as $C \to +\infty$.
\end{proof}

 The geometric appearance of the orthogonal manifolds in low-dimensional cases is visualized in Fig.~\ref{fig:orthogonal_manifolds}. In the transitional orthants (Fig.~\ref{fig:orthogonal_manifolds}b and \ref{fig:orthogonal_manifolds}c), the mixed signs of the potential field coefficients force the limit cases $C \to -\infty$ and $C \to +\infty$ onto distinct, non-adjacent boundary hyperplanes. This forms distinct `entry' and `exit' portals, allowing the kinetic fibers to fully traverse the orthant. In the extremal orthants shown in Fig.~\ref{fig:orthogonal_manifolds}a and \ref{fig:orthogonal_manifolds}d, all the hyperplanes bounding the orthant are part of a unique entry portal to which the manifolds converge as $C \to -\infty$. In this case, the manifolds foliate the full orthant, sweeping through it in the direction of the central fiber as $C \to +\infty$.

 \begin{figure*}[t]
     \centering
     \begin{subfigure}[b]{0.22\linewidth}
         \centering
         \includegraphics[width=\linewidth]{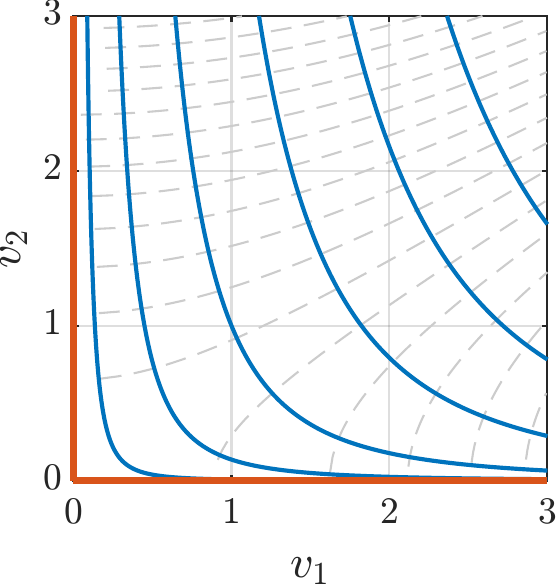}
         \caption{2D Q-I (Extremal)}
         \label{subfig:manifold_2d_q1_ext}
     \end{subfigure}
     \hfill
     \begin{subfigure}[b]{0.22\linewidth}
         \centering
         \includegraphics[width=\linewidth]{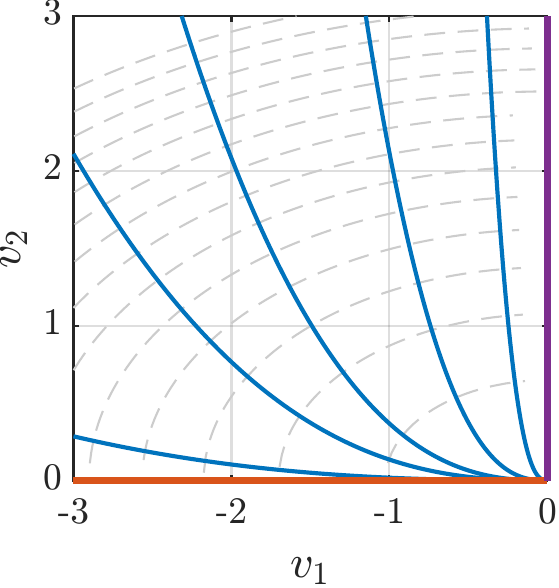}
         \caption{2D Q-II (Transitional)}
         \label{subfig:manifold_2d_q2_trans}
     \end{subfigure}
     \hfill
     \begin{subfigure}[b]{0.26\linewidth}
         \centering
         \includegraphics[width=\linewidth]{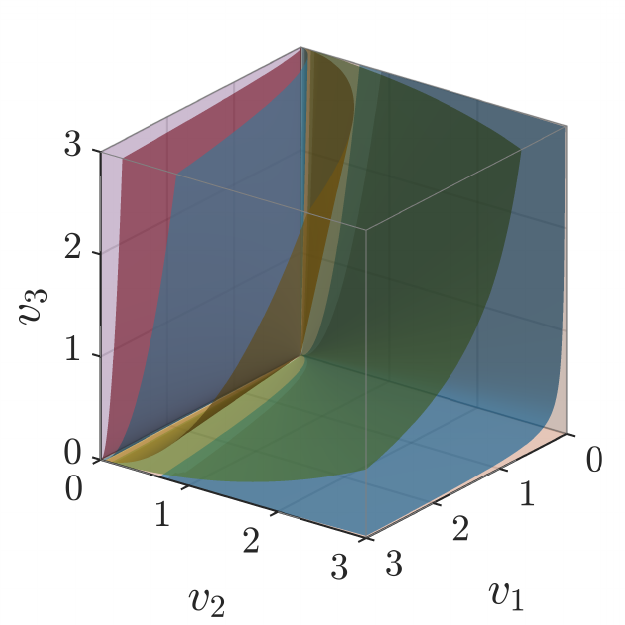}
         \caption{3D (+,+,+) (Transitional)}
         \label{subfig:manifold_3d_trans}
     \end{subfigure}
     \hfill
     \begin{subfigure}[b]{0.26\linewidth}
         \centering
         \includegraphics[width=\linewidth]{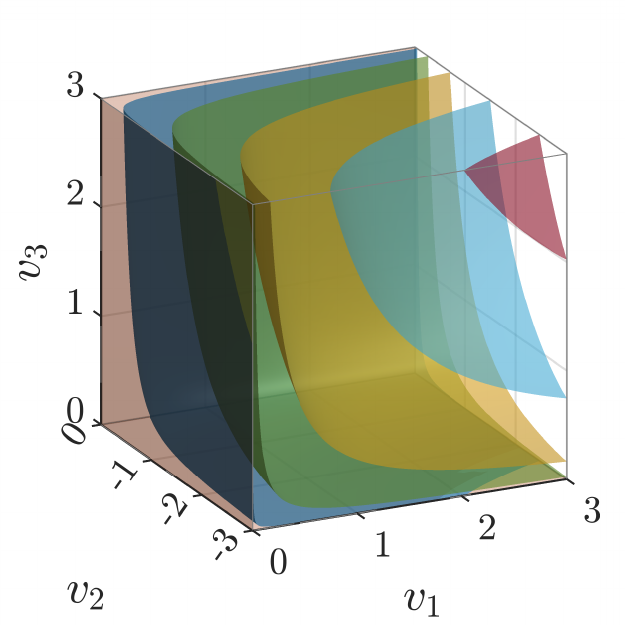}
         \caption{3D (+,-,+) (Extremal)}
         \label{subfig:manifold_3d_ext}
     \end{subfigure}
    
     \caption{Level sets of the global logarithmic potential field defining the orthogonal manifolds. The top row utilizes the $n=2$ asymmetric allocation map from Fig.~\ref{fig:task_fibers}b, showing the family of curves for $C \in \{-10, -4, 0, 3 5, 6.5\}$ in the extremal Quadrant I and $C \in \{-8, -4, -2, 1.5, 7\}$ in the transitional Quadrant II. The bottom row utilizes the $n=3$ standard allocation from Fig.~\ref{fig:task_fibers}c, illustrating the surfaces for $C \in \{ -5,-1.5, 0.5, 2.5, 5\}$  in the transitional $(+,+,+)$ octant and $C \in \{ -4,-0.5, 1.5, 3, 4\}$ in the extremal $(+,-,+)$ octant. Thick colored boundaries denote the limit cases $C \to \pm\infty$, highlighting the difference between entry and exit portals in the transitional and extremal orthants. When present, background dashed lines indicate the fibers.}
     \label{fig:orthogonal_manifolds}
 \end{figure*}

\section{Layer Partition of the Orthants and Global Orthogonal Sections}
\label{sec:layer_partition}

\subsection{The Stratified Layer Partition of Orthants}
\label{subsec:layer_partition}

The global topology of the task fibers implies a highly structured partitioning of the transitional orthants in the $n$-dimensional kinetic space $\mathcal{V}$ which are organized into distinct layers based on their progression between the two extremal orthants. 
To formalize this stratification, consider the transformed coordinate mapping of a generic fiber: $x_k(\lambda) = (z_w^A)_k + \lambda b_k^A$. Because this function is strictly linear with respect to $\lambda$, each coordinate $x_k$ must cross zero exactly once as $\lambda$ sweeps from $-\infty$ to $+\infty$. Consequently, every coordinate changes its sign relative to the null-space vector exactly once along any valid fiber. Traveling along a continuous fiber from $\lambda \to -\infty$ to $\lambda \to +\infty$ is therefore combinatorially equivalent to starting with a kinetic state entirely opposite in sign to $b^A$ and flipping the signs one by one until they are entirely aligned with $b^A$.
This strict monotonic progression allows the $2^n$ orthants to be grouped by the number of signs that have successfully transitioned. The initial negative extremal orthant $\mathcal{O}^{-}$ possesses $n$ exit faces and $0$ entry faces. As a fiber crosses a boundary face into a subsequent orthant, exactly one exit face transitions into an entry face. 

This structural grouping dictates the geometric boundaries shared by orthants within the same progressive stage. If two distinct orthants $\mathcal{O}_{\sigma_a}$ and $\mathcal{O}_{\sigma_b}$ possess the exact same number of entry boundary faces, 
i.e., 
$|\mathcal{I}^{in}_{\sigma_a}| = |\mathcal{I}^{in}_{\sigma_b}|$, their signatures $\sigma_a$ and $\sigma_b$ cannot differ by just a single sign; changing one sign would alter their entry-face count, placing them in different stages of the fiber progression. To maintain the same count, $\sigma_a$ and $\sigma_b$ must differ by exactly two indices: one strictly transitioning from entry to exit, and the other from exit to entry. Because these two orthants differ by two signs rather than one, they bypass each other's $(n-1)$-dimensional flat faces entirely. Instead, they intersect exclusively at an $(n-2)$-dimensional bounding stratum (e.g., an edge in $\mathbb{R}^3$). Because the two coordinates defining this intersection consist of an entry/exit pair, this bounding stratum matches the geometric definition of a hinge. This leads us to the following definition.

\begin{defn}[Orthant Layers and Reciprocal Hinges]
\label{def:traversal_layers}
The kinetic space orthants are partitioned into \emph{orthant layers} $\mathcal{L}_l$, where the integer $l \in \{0, \dots, n\}$ represents the exact number of acquired entry faces, $\mathcal{L}_l=\bigcup_{|\mathcal{I}^{in}_{\sigma}|=l}\{\mathcal{O}_\sigma\}$.
Layer $\mathcal{L}_0 = \{\mathcal{O}^-\}$ is the negative extremal orthant, and $\mathcal{L}_n = \{\mathcal{O}^+\}$ is the positive extremal orthant. The $2^n-2$ transitional orthants form the intermediate layers $l \in \{1, \dots, n-1\}$, with each layer containing exactly $\binom{n}{l}$ distinct orthants. 

For any two distinct orthants $\mathcal{O}_{\sigma_a}, \mathcal{O}_{\sigma_b} \in \mathcal{L}_l$ whose sign vectors $\sigma_a$ and $\sigma_b$ differ by exactly two indices, their shared boundary is strictly an $(n-2)$-dimensional stratum. This exclusive intersection constitutes a \emph{reciprocal hinge}, acting as an orthant hinge simultaneously for both $\mathcal{O}_{\sigma_a}$ and $\mathcal{O}_{\sigma_b}$.
\end{defn}

\begin{thm}[Global Fiber Capture by Stratified Layers]
\label{thm:global_capture}
For an $n$-dimensional kinetic space $\mathcal{V}$, any single complete layer $\mathcal{L}_l$ acts as a complete transverse section to the fiber flow, sufficient to intersect the fibers $f^{-1}(w)$ for all $w \in \mathcal{W}$. The task-space pre-images of the $\binom{n}{l}$ orthants within a single layer $\mathcal{L}_l$ possess strictly disjoint interiors---separated by the tangent boundary fibers intersecting the reciprocal hinges---and their union completely partitions the global set of achievable tasks $w \in \mathcal{W}$. 
\end{thm}
\begin{proof}
Every continuous reconfiguration fiber originating from $\mathcal{O}^{-}$ (layer $\mathcal{L}_0$) and terminating at $\mathcal{O}^{+}$ (layer $\mathcal{L}_n$) must strictly and monotonically accumulate entry coordinates, because each $x_k(\lambda)$ changes sign exactly once. Therefore, every fiber traversing the full kinetic sequence must intersect exactly one orthant in layer $\mathcal{L}_1$, one in $\mathcal{L}_2$, up through $\mathcal{L}_{n-1}$, unless it geometrically intersects a reciprocal hinge. Because fibers intersecting reciprocal hinges are purely tangent (Theorem~\ref{thm:topological_sweep}) and define the exact geometric boundaries between the disjoint orthant interiors, the union of the closures of the $\binom{n}{l}$ orthants in any single layer $\mathcal{L}_l$ forms a contiguous, complete transverse net. Consequently, the task-space pre-images of the orthants in $\mathcal{L}_l$ seamlessly piece together to form the entire task space $\mathcal{W}$.
\end{proof}

\begin{prop}[Combinatorics of Reciprocal Hinges]
\label{prop:hinge_combinatorics}
The exact number of unique reciprocal hinges contained within a specific orthant layer $\mathcal{L}_l$ is given by:
\begin{equation}
    N_{\mathcal{H}}(l) = \frac{1}{2} \binom{n}{l} l (n-l) = \binom{n}{2}\binom{n-2}{l-1}.
    \label{eq:number_hinges_in_a_layer}
\end{equation}
\end{prop}
\begin{proof}
The reciprocal hinges geometrically glue together adjacent orthants within the same layer $\mathcal{L}_l$. To transition between two adjacent orthants while remaining strictly within $\mathcal{L}_l$ (maintaining exactly $l$ coordinates of a given sign signature), exactly two coordinates must simultaneously vanish and swap signs: one from the $l$ variables of the given sign signature, and one from the $n-l$ variables of the opposite signature. Therefore, from the perspective of any single orthant in $\mathcal{L}_l$, there are exactly $l(n-l)$ possible pairs of coordinates that can form an intra-layer reciprocal hinge. 

Summing these possible hinges across all $\binom{n}{l}$ orthants in the layer counts each hinge twice, because every intra-layer reciprocal hinge is a boundary shared by exactly two adjacent orthants. Dividing by two yields the exact total $N_{\mathcal{H}}(l)$ in~\eqref{eq:number_hinges_in_a_layer}.
\end{proof}

Because the binomial coefficient $\binom{n-2}{l-1}$ is maximized at the central values of $l$, Proposition~\ref{prop:hinge_combinatorics} proves analytically that the deepest transitional layers (centered around $l \approx n/2$) contain the maximum number of reciprocal hinges and, consequently, the highest density of kinematic singularities.

\subsection{Global Sections in Transitional Layers}
\label{subsec:global_sections_transitional_layers}

The establishment of transitional layers $\mathcal{L}_l$ provides the exact topological framework required to synthesize the local orthogonal manifolds of each orthant into a global structure. As proven in Theorem~\ref{thm:global_capture}, the orthants within a single layer $\mathcal{L}_l$ collectively capture the entire transversal flow of the reconfiguration fibers, with their boundaries contiguous at the reciprocal hinges. 

Recall that within any single transitional orthant $\mathcal{O}_\sigma$, the orthogonal manifold $\mathcal{M}_{C,\sigma}$ defined by the level set of the potential function $\Phi_\sigma(v) = C$ acts as a complete transversal barrier to the local fiber flow within that orthant, anchored at the hinges for any $C\in\mathbb{R}$. This localized property scales globally when constrained to a specific layer.

\begin{defn}[Global Orthogonal Section]
\label{def:global_section}
For a given transitional layer $\mathcal{L}_l$ and a fixed potential constant $C \in \mathbb{R}$, we define the \emph{global orthogonal section} $\mathcal{M}^{(l)}_C$ as the union of the closures of the localized orthogonal manifolds residing in each orthant of that layer:
\begin{equation}
    \mathcal{M}^{(l)}_C = \bigcup_{\sigma \in \mathcal{L}_l} \bar{\mathcal{M}}_{C,\sigma}.
\end{equation}
Topologically, $\mathcal{M}^{(l)}_C$ resembles a multi-petaled structure, where each localized manifold $\mathcal{M}_{C,\sigma}$ serves as a petal. Because the orthants in $\mathcal{L}_l$ are structurally glued at their reciprocal hinges, and because the logarithmic potential $\Phi_\sigma(v) \to C$ remains consistent at these boundaries, the individual petals are continuously joined at the $(n-2)$-dimensional reciprocal hinges.
\end{defn}

Because the layer $\mathcal{L}_l$ completely partitions the achievable task space $\mathcal{W}$, and each petal $\mathcal{M}_{C,\sigma}$ strictly transverses the fibers of its respective orthant, the orthogonal section $\mathcal{M}^{(l)}_C$ constitutes a complete, global orthogonal section of the entire task-space fiber bundle. Every fiber intersects $\mathcal{M}^{(l)}_C$ orthogonally and exactly once.

\begin{thm}[$C^1$ Smoothness of Global Sections]
\label{thm:section_smoothness}
For any transitional layer $\mathcal{L}_l$ ($l \in \{1, \dots, n-1\}$) and any constant $C \in \mathbb{R}$, the global orthogonal section $\mathcal{M}^{(l)}_C$ is a $C^1$ smooth differential manifold across its interior. Specifically, at any reciprocal hinge joining two adjacent orthants $\mathcal{O}_A, \mathcal{O}_B \in \mathcal{L}_l$, the tangent hyperplanes of the respective local manifolds $\mathcal{M}_{C,A}$ and $\mathcal{M}_{C,B}$ continuously align.
\end{thm}
\begin{proof}
Let $p$ be a point strictly residing on a reciprocal hinge that forms the shared boundary between two adjacent orthants $\mathcal{O}_A$ and $\mathcal{O}_B$ within $\mathcal{L}_l$. By Definition~\ref{def:traversal_layers}, $\mathcal{O}_A$ and $\mathcal{O}_B$ differ in exactly two coordinate signs. Let these indices be $i$ and $j$, such that $\sigma_i^B = -\sigma_i^A$ and $\sigma_j^B = -\sigma_j^A$. 

Because $p$ lies on this reciprocal hinge, the kinetic coordinates approach zero exactly at these indices: $v_i \to 0$ and $v_j \to 0$, while all other spatial coordinates $v_k$ (for $k \neq i,j$) remain strictly non-zero and finite. To simplify notation, let $\alpha = \sigma_i^A b_i^A$ and $\beta = \sigma_j^A b_j^A$. Because one index represents an entry face and the other an exit face, $\alpha$ and $\beta$ possess opposite signs. Without loss of generality, assume $\alpha > 0$ and $\beta < 0$, and define $\gamma = -\beta > 0$.

Consider the normal vector to the petal $\mathcal{M}_{C,A}$ within the interior of $\mathcal{O}_A$, given by the gradient of the potential function:
$$
    \nabla \Phi_A(v) = \left[ \dots, \frac{\alpha}{v_i}, \dots, \frac{-\gamma}{v_j}, \dots \right]^T.
$$
As $v \to p$, the components $i$ and $j$ diverge to infinity, while all other components remain finite. Consequently, the unit normal vector $\hat{n}_A = \frac{\nabla \Phi_A}{\|\nabla \Phi_A\|}$ collapses strictly into the two-dimensional $v_i-v_j$ plane. The geometric orientation of this normal vector within that plane is defined by the ratio of its components:
$$
    R_A = \frac{(\nabla \Phi_A)_j}{(\nabla \Phi_A)_i} = -\frac{\gamma}{\alpha} \left( \frac{v_i}{v_j} \right).
$$
The asymptotic behavior of the ratio $(v_i/v_j)$ is governed by the level set constraint $\Phi_A(v) = C$. Let $S(p)$ be the finite sum of the remaining potential terms evaluated at $p$. The constraint requires $\alpha \ln(v_i) - \gamma \ln(v_j) + S(p) = C$. Solving for $v_i$ yields $v_i = v_j^{\gamma/\alpha} e^{(C-S(p))/\alpha}$. Substituting this into the ratio $R_A$ gives:
$$
    R_A = -\frac{\gamma}{\alpha} v_j^{\left(\frac{\gamma}{\alpha} - 1\right)} e^{\frac{C-S(p)}{\alpha}}.
$$
The limit of $R_A$ as $v_j \to 0$ depends strictly on the system parameters $\alpha$ and $\gamma$:
\begin{itemize}
    \item \emph{Asymmetric Case I ($\gamma > \alpha$):} The exponent $(\gamma/\alpha - 1)$ is positive. As $v_j \to 0$, $R_A \to 0$. The normal vector fully aligns with the $v_j$ axis, independent of $C$.
    \item \emph{Asymmetric Case II ($\gamma < \alpha$):} The exponent is negative. As $v_j \to 0$, $R_A \to -\infty$. The normal vector fully aligns with the $v_i$ axis, independent of $C$.
    \item \emph{Symmetric Case ($\gamma = \alpha$):} The exponent is zero. As $v_j \to 0$, $R_A \to -e^{(C-S(p))/\alpha}$. The normal vector converges to a specific diagonal angle dictated by the level set constant $C$.
\end{itemize}

Now consider the adjacent petal $\mathcal{M}_{C,B}$ within $\mathcal{O}_B$. Its potential gradient components for $i$ and $j$ negate exactly: $(\nabla \Phi_B)_i = -\alpha/v_i$ and $(\nabla \Phi_B)_j = \gamma/v_j$. The planar ratio evaluates to $R_B = \frac{\gamma/v_j}{-\alpha/v_i} \equiv R_A$. Because the finite remainder $S(p)$ relies only on the $n-2$ coordinates whose signs are identical across both orthants, $S(p)$ is invariant. Because the level-set constant $C$ is identically enforced, the limit of $R_B$ evaluates to the exact same geometrically parallel normal direction as $R_A$ across all three parametric cases.

Because the dominant, infinite components of $\nabla \Phi_B$ are the exact algebraic negations of $\nabla \Phi_A$, and their planar ratios match perfectly, their unit normal vectors in the limit satisfy $\lim_{v \to p} \hat{n}_B(v) = - \lim_{v \to p} \hat{n}_A(v)$. A unit normal vector and its exact negation define the identical unoriented one-dimensional normal space. Their orthogonal complements---the tangent hyperplanes $T_p\mathcal{M}_{C,A}$ and $T_p\mathcal{M}_{C,B}$---are therefore geometrically identical. Thus, the global section $\mathcal{M}^{(l)}_C$ connects smoothly at the reciprocal hinges, guaranteeing $C^1$ continuity.
\end{proof}

 \begin{rem}[Topological Necessity of a Global Constant $C$]
 The proof of Theorem~\ref{thm:section_smoothness} reveals a geometric nuance: in the asymmetric cases ($|\alpha| \neq |\beta|$), the explosive divergence of the gradient forces the tangent plane to align with the bounding axes independently of the level set constant $C$. Geometrically, this implies that one could artificially construct a piecewise manifold by gluing a local petal $\mathcal{M}_{C_A,A}$ to a neighboring petal $\mathcal{M}_{C_B,B}$ (where $C_A \neq C_B$), and in asymmetric conditions, the resulting surface might still retain $C^1$ smoothness.

 However, enforcing a uniform global constant $C$ across the entire layer $\mathcal{L}_l$ is topologically mandatory. First, for parametric robustness, as demonstrated in the symmetric case ($|\alpha| = |\beta|$), if the kinetic parameters of the crossing variables are balanced, the tangent angle becomes strictly reliant on $C$. Gluing disparate $C$ values in this regime would result in a sharp transversal kink, destroying $C^1$ continuity. Second, the global manifold $\mathcal{M}^{(l)}_C$ represents an isopotential barrier to the transversal fiber flow. Enforcing $C_A = C_B = C$ preserves the continuous physical dynamics of the configuration space. Sweeping the single constant $C \in (-\infty, \infty)$ ensures that the global manifolds cleanly foliate the entire transitional layer $\mathcal{L}_l$, creating a continuous, non-intersecting (despite being tangent at the reciprocal hinges), one-to-one mapping of the global task space.
 \end{rem}

A visual demonstration of some of the main concepts introduced in this section is shown in Figs.~\ref{fig:2d_cases} and \ref{fig:3d_cases}. While these figures refer unavoidably to the low-dimensional cases $n=2$ and $n=3$, they are instrumental for building an intuition for the geometry developed so far, which remains valid in any higher dimension $n\in \mathbb{N}$.

 \begin{figure}[t]
     \centering
     \hfill\begin{subfigure}{0.49\columnwidth}
         \centering
         \includegraphics[width=0.99\linewidth]{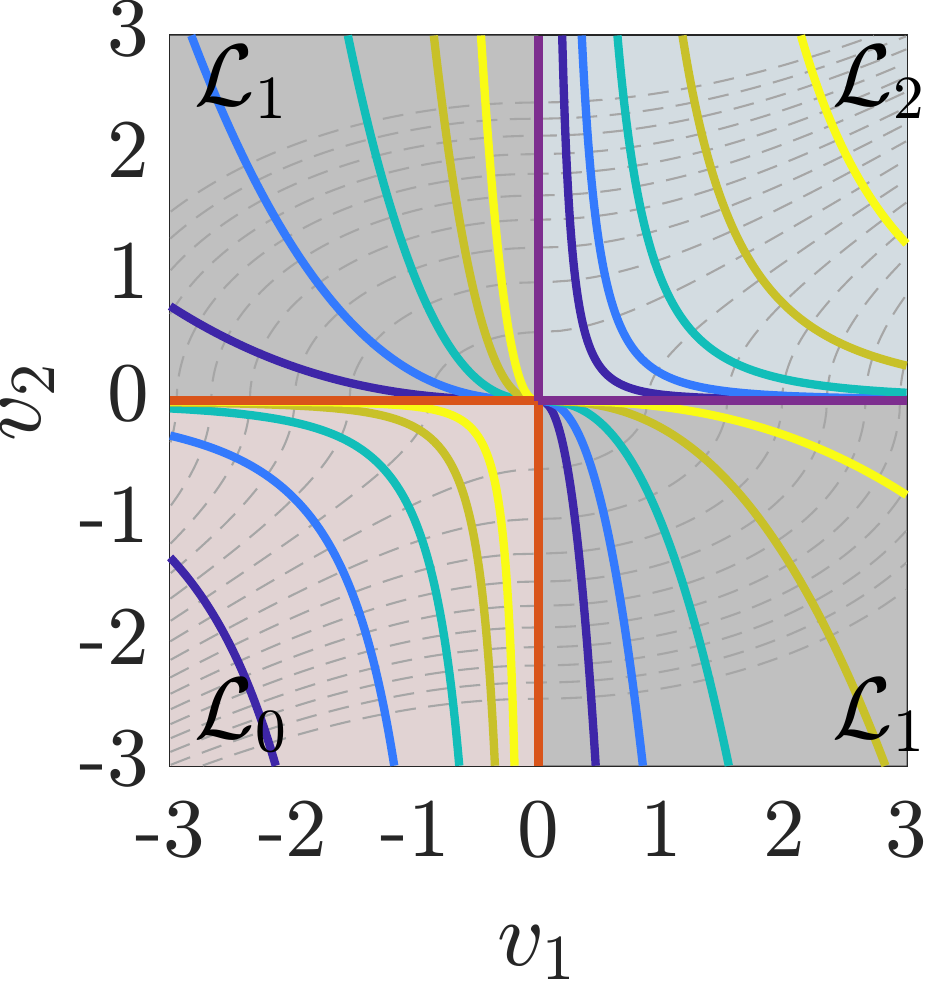}
         \caption{Asymmetr. Case ($|\alpha| \neq |\beta|$).}
         \label{fig:2d_asym}
     \end{subfigure}\hfill
     \begin{subfigure}{0.49\columnwidth}
         \centering
         \includegraphics[width=0.99\linewidth]{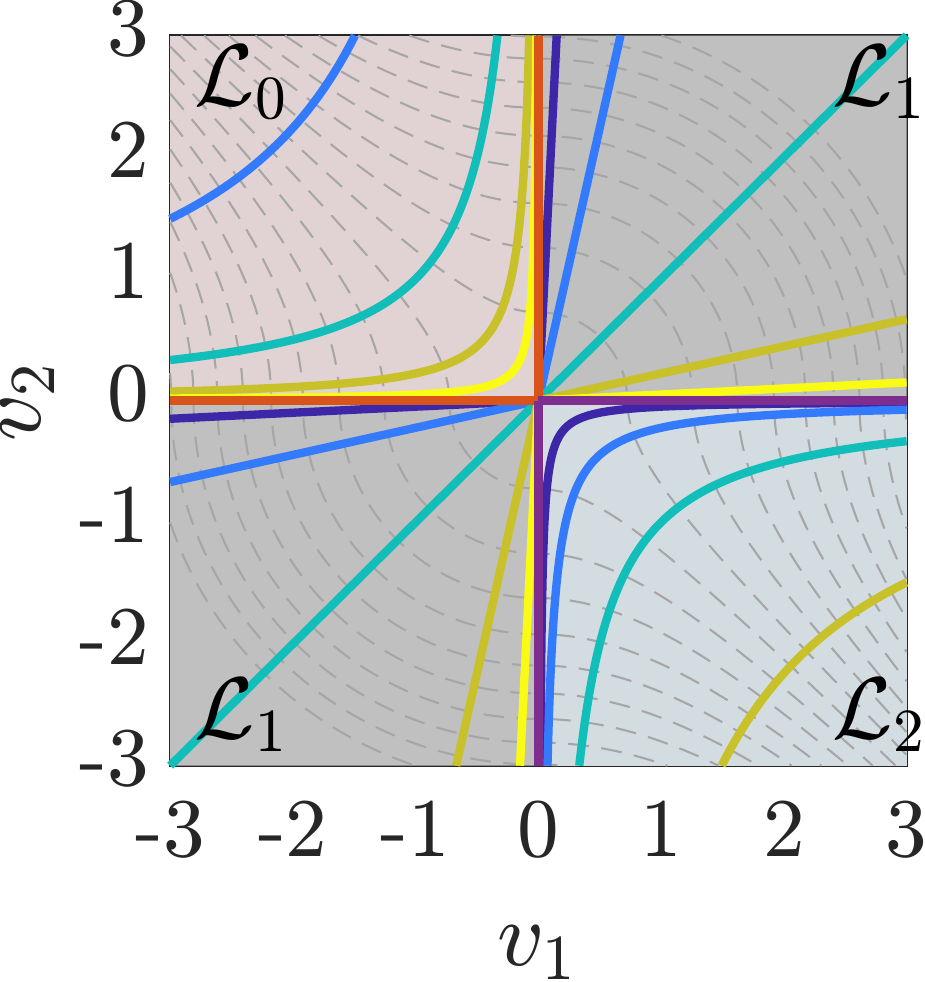}
         \caption{Balanced Case ($|\alpha| = |\beta|$).}
         \label{fig:2d_bal}
     \end{subfigure}\hfill~
    
     \caption{Global orthogonal sections and layer partitioning in a 2D kinetic space ($n=2$). The phase space is stratified into topological layers $\mathcal{L}_0, \mathcal{L}_1,$ and $\mathcal{L}_2$ (background shading). In $\mathcal{L}_1$, the fibers (dashed gray) flow strictly from the entry portals (orange) to the exit portals (purple). The global sections $\mathcal{M}^{(1)}_C$ act as complete orthogonal transversal barriers, seamlessly gluing the local orthant petals across the reciprocal hinge (the origin) to form a $C^1$ smooth global manifold.}
     \label{fig:2d_cases}
 \end{figure}

 \begin{figure*}[t]
     \centering
    
     \begin{subfigure}{0.24\textwidth}
         \centering
         \includegraphics[width=\linewidth]{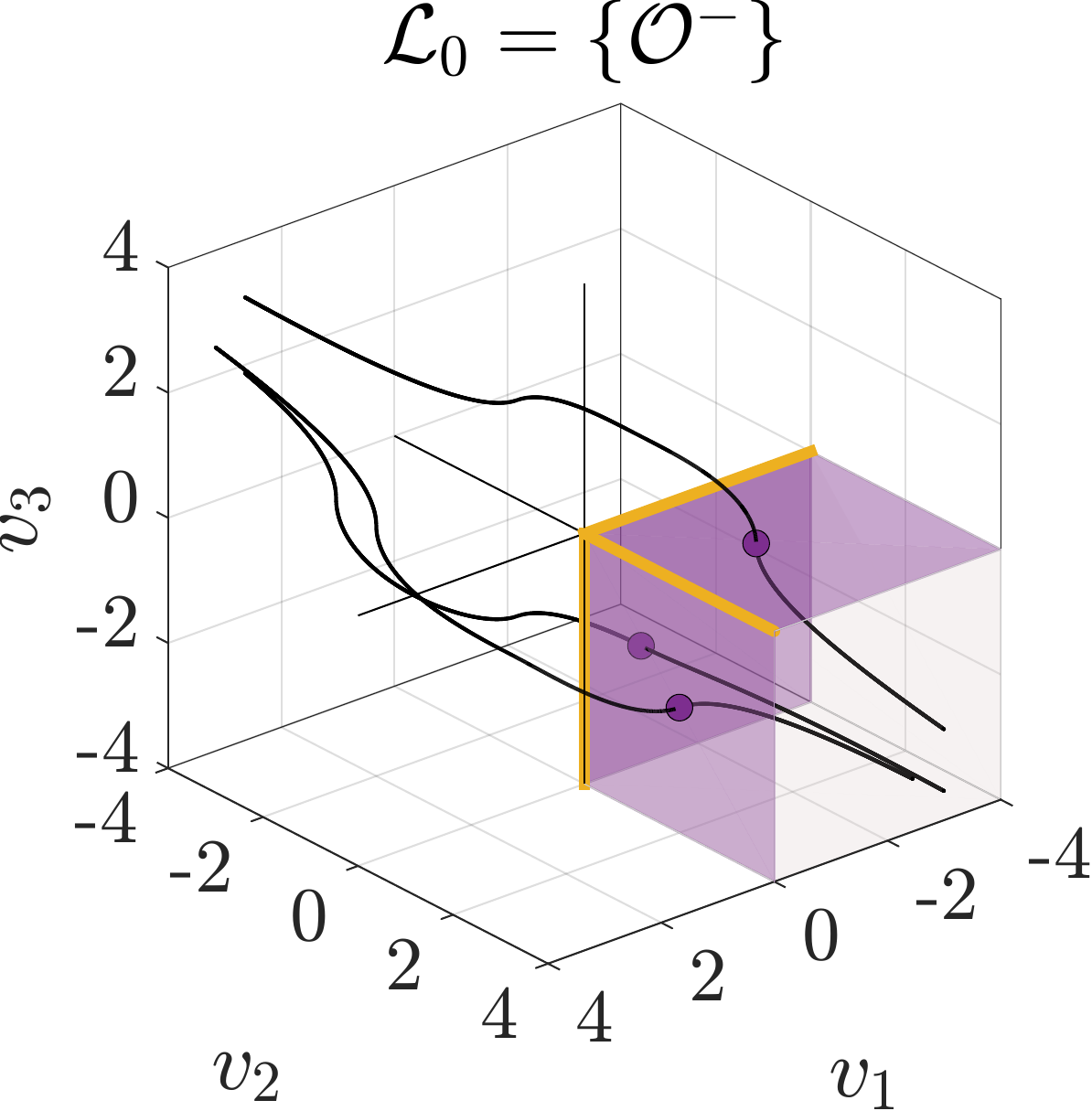}
         \caption{$\mathcal{L}_0$}
         \label{fig:3d_l0}
     \end{subfigure}\hfill
     \begin{subfigure}{0.24\textwidth}
         \centering
         \includegraphics[width=\linewidth]{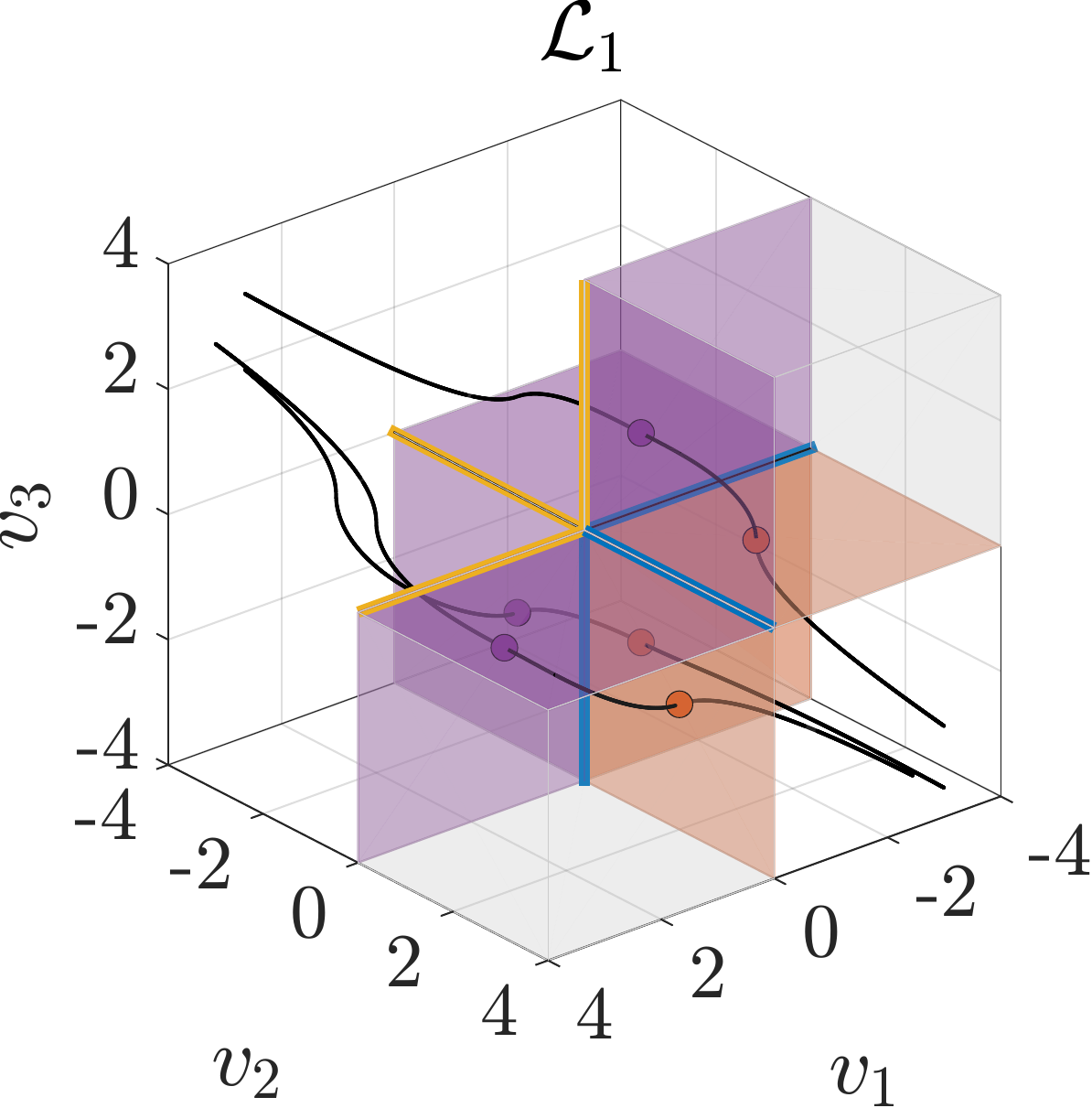}
         \caption{$\mathcal{L}_1$}
         \label{fig:3d_l1}
     \end{subfigure}\hfill
     \begin{subfigure}{0.24\textwidth}
         \centering
         \includegraphics[width=\linewidth]{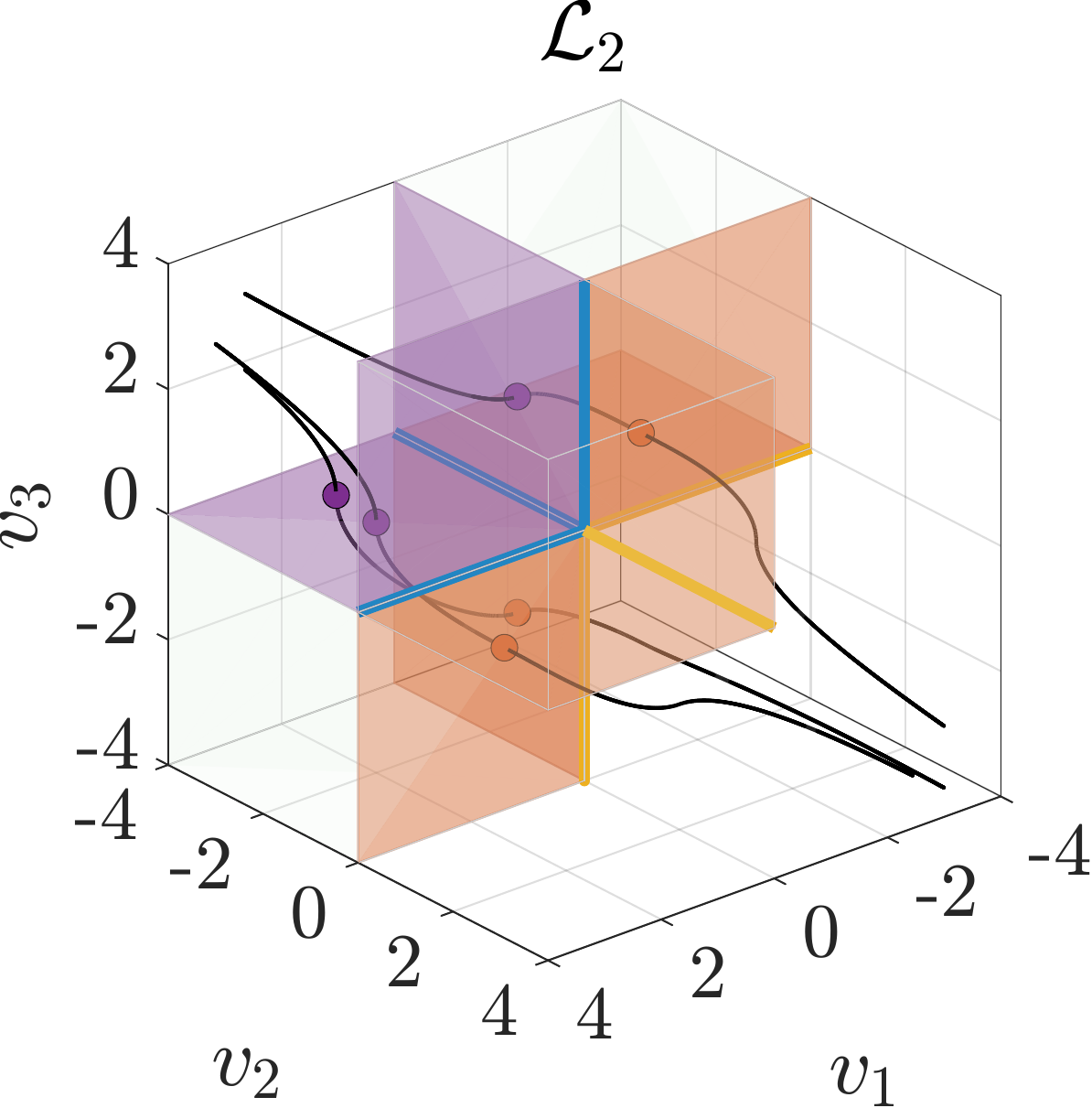}
         \caption{$\mathcal{L}_2$}
         \label{fig:3d_l2}
     \end{subfigure}\hfill
     \begin{subfigure}{0.24\textwidth}
         \centering
         \includegraphics[width=\linewidth]{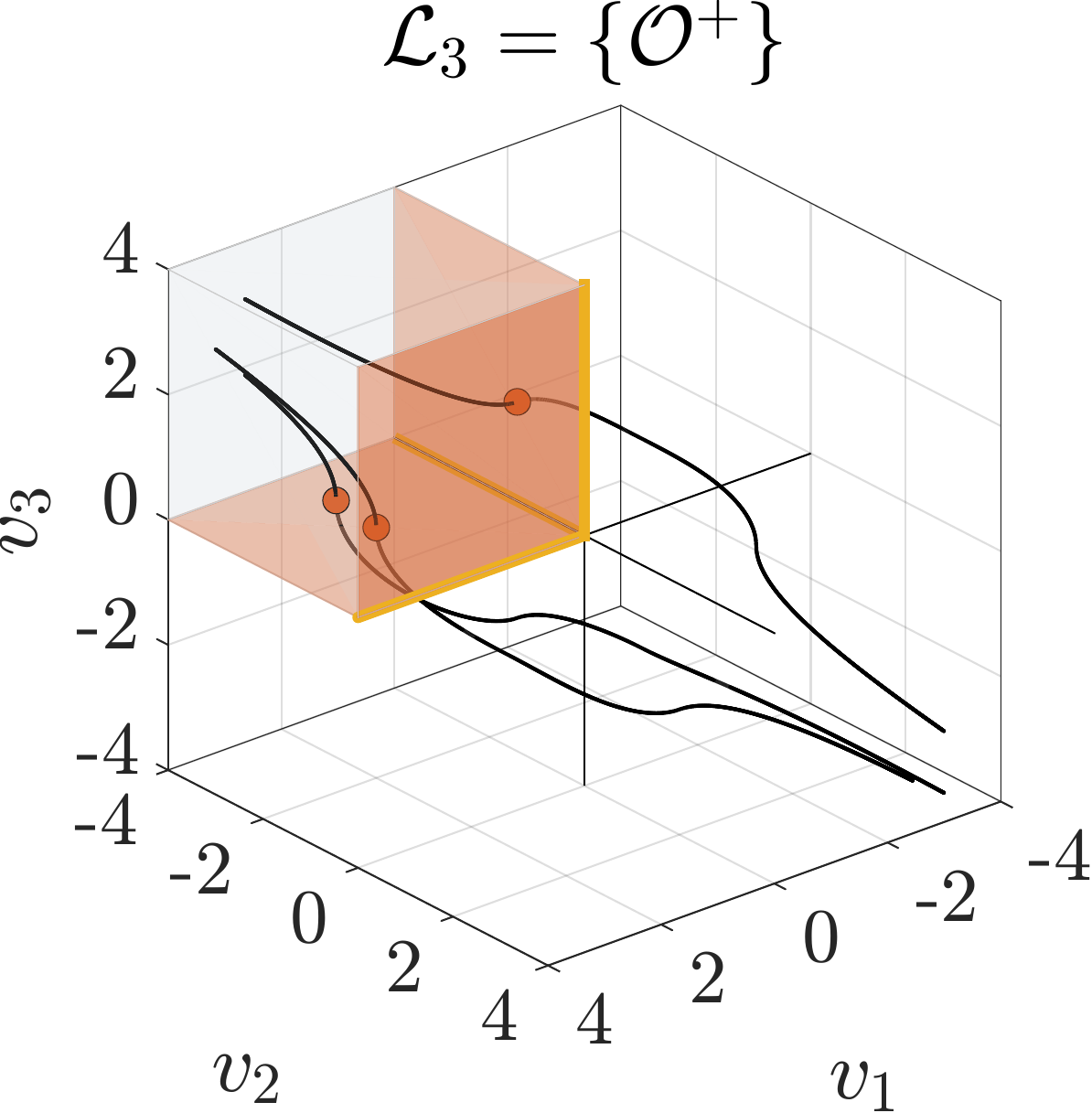}
         \caption{$\mathcal{L}_3$}
         \label{fig:3d_l3}
     \end{subfigure}
    
     \vspace{0.5cm} 
    
     \begin{subfigure}{0.24\textwidth}
         \centering
         \includegraphics[width=\linewidth]{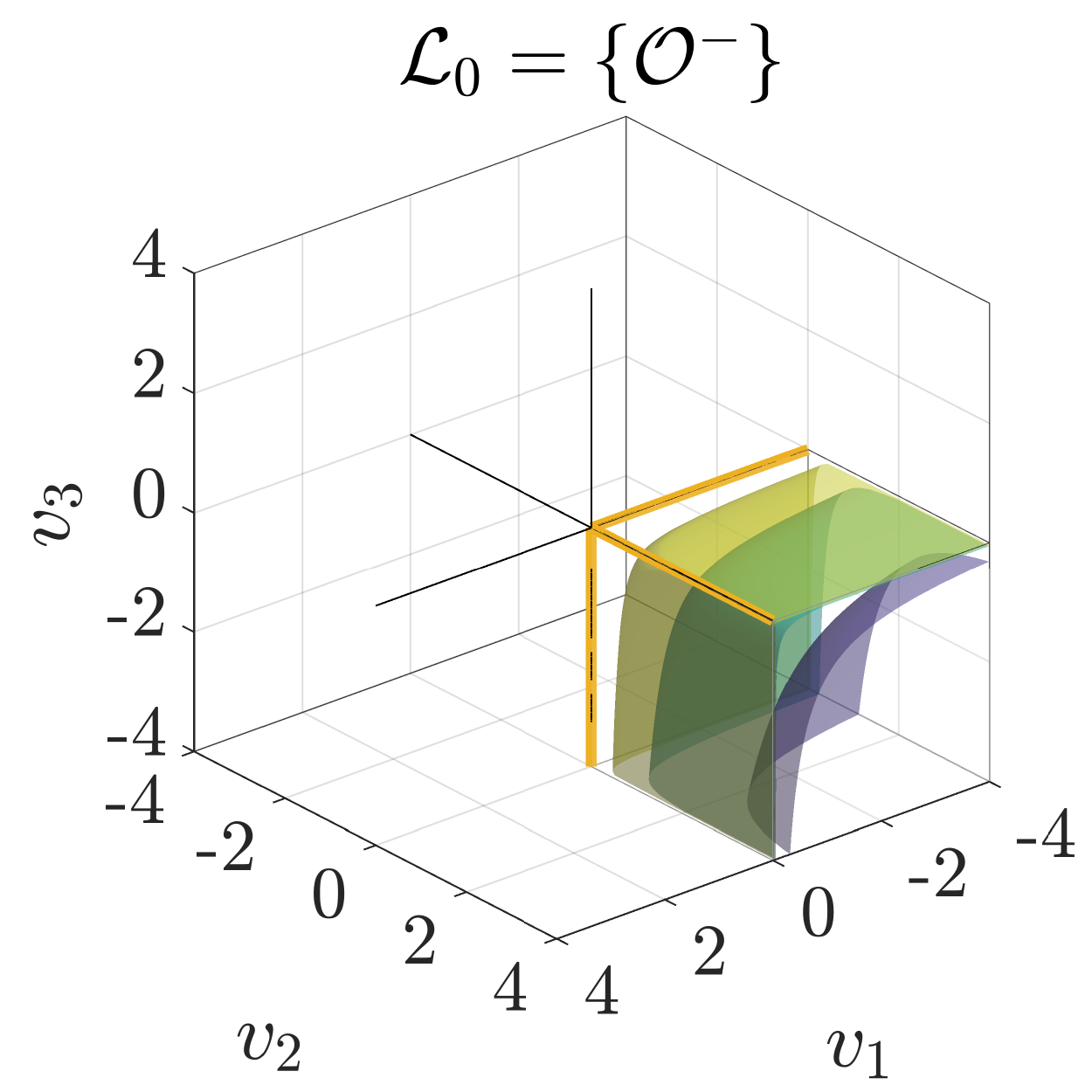}
         \caption{Sections in $\mathcal{L}_0$}
         \label{fig:3d_sec_l0}
     \end{subfigure}\hfill
     \begin{subfigure}{0.24\textwidth}
         \centering
         \includegraphics[width=\linewidth]{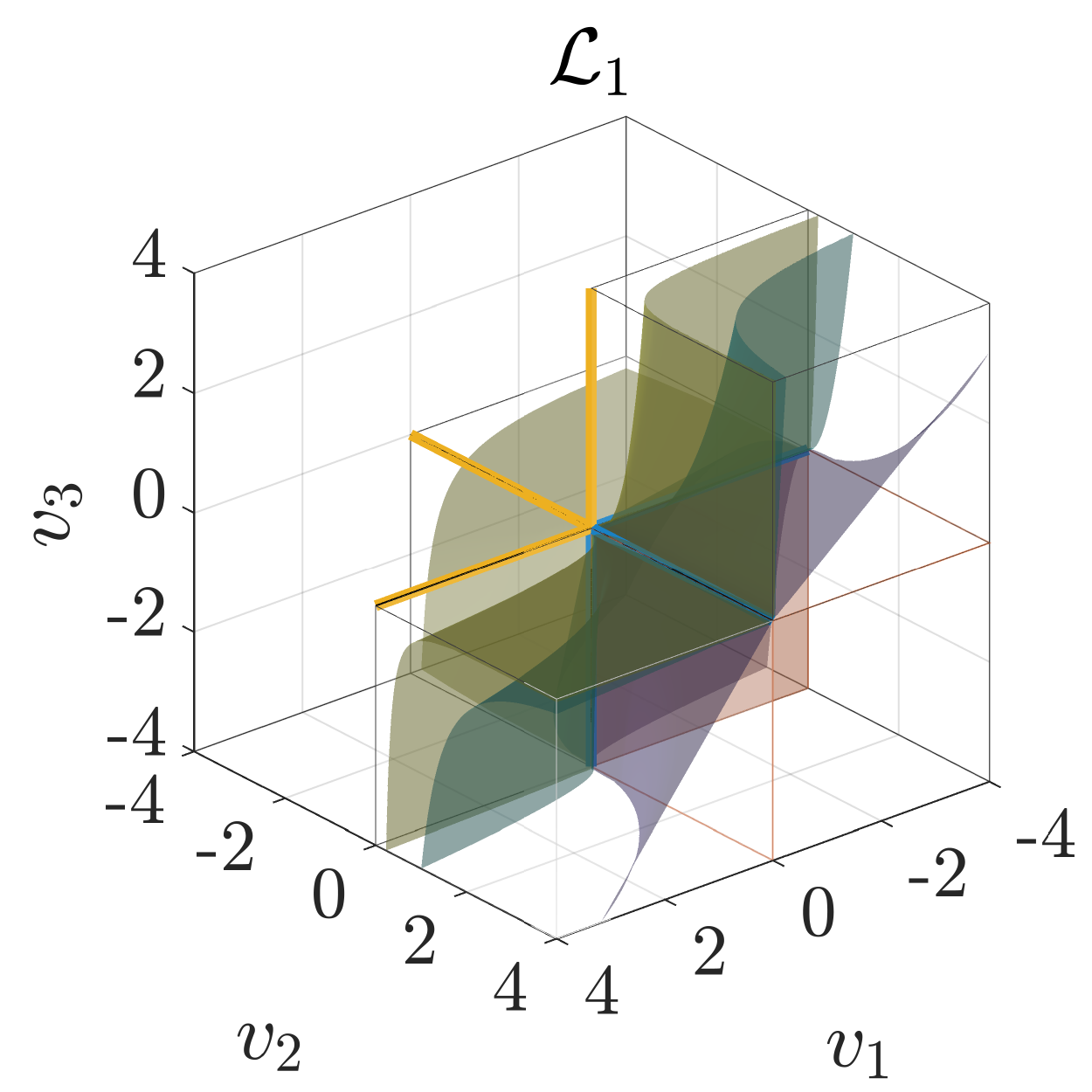}
         \caption{Sections in $\mathcal{L}_1$}
         \label{fig:3d_sec_l1}
     \end{subfigure}\hfill
     \begin{subfigure}{0.24\textwidth}
         \centering
         \includegraphics[width=\linewidth]{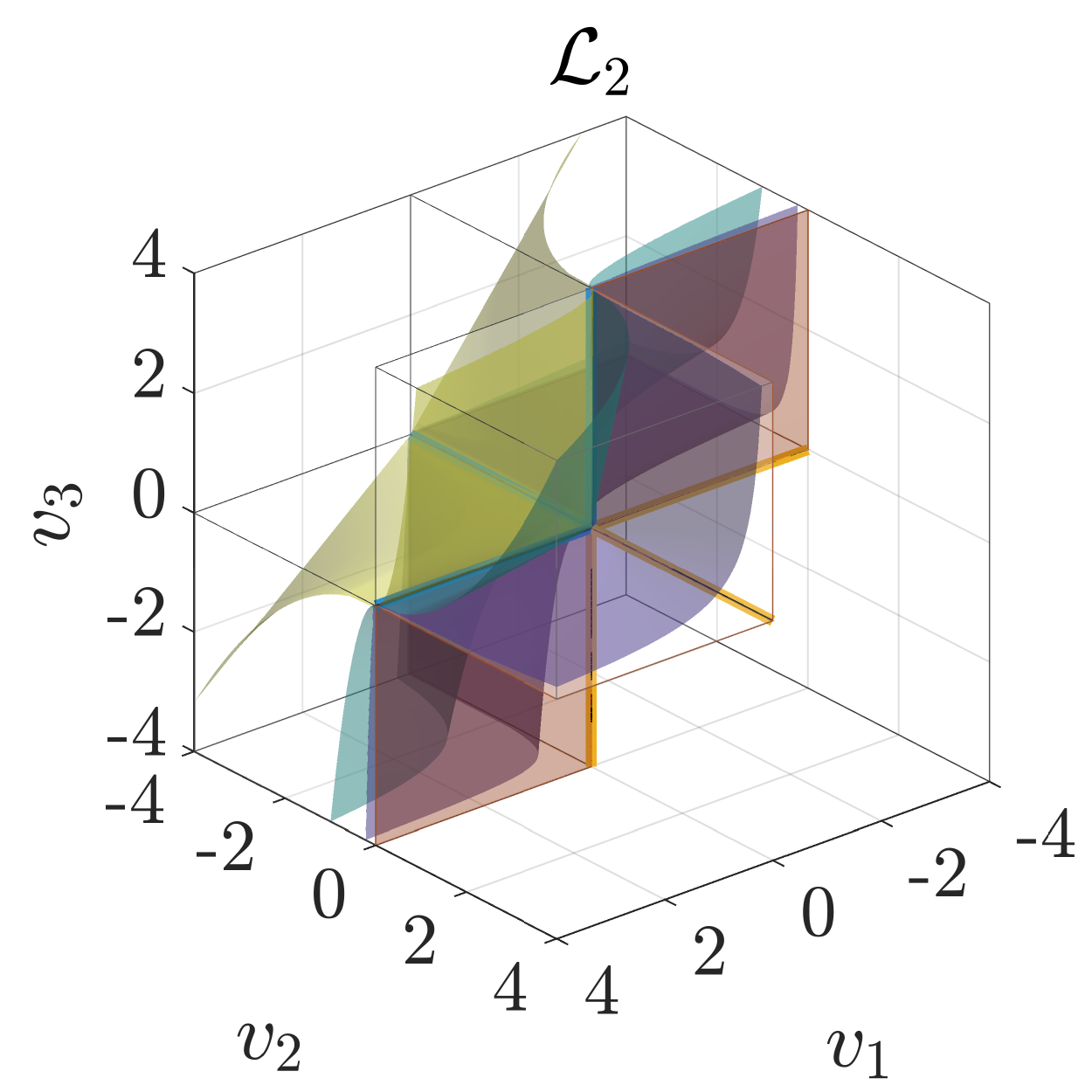}
         \caption{Sections in $\mathcal{L}_2$}
         \label{fig:3d_sec_l2}
     \end{subfigure}\hfill
     \begin{subfigure}{0.24\textwidth}
         \centering
         \includegraphics[width=\linewidth]{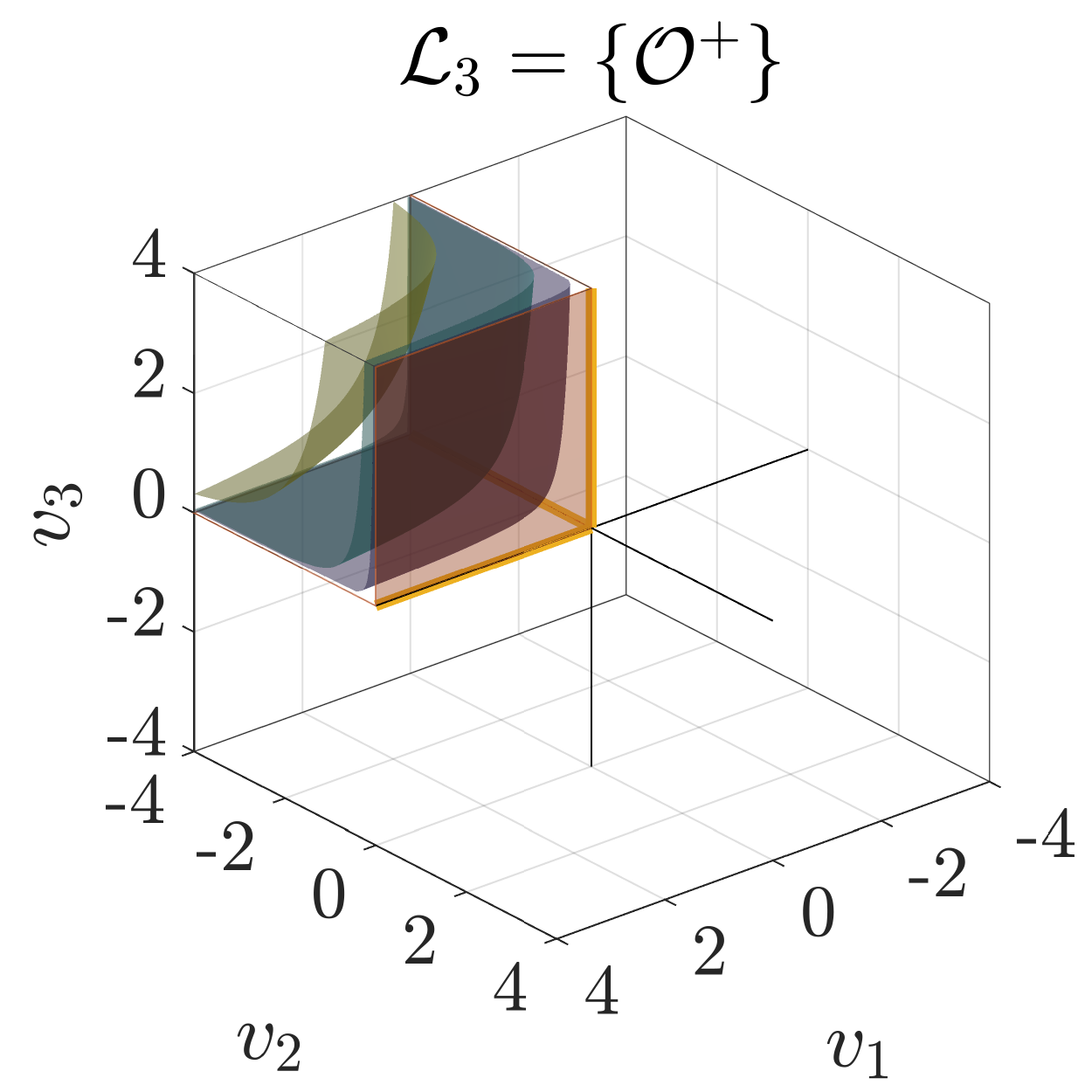}
         \caption{Sections in $\mathcal{L}_3$}
         \label{fig:3d_sec_l3}
     \end{subfigure}
     
     \caption{Global stratification and orthogonal sections of a 3D kinetic space ($n=3$). (Top) The sequential layer partition from the negative extremal layer $\mathcal{L}_0$ to the positive extremal layer $\mathcal{L}_3$. Exit portals are indicated in purple, entry portals in orange, folds by blue semi-axes, and hinges by gold semi-axes. Representative reconfiguration fibers (black curves) strictly traverse this sequence, entering at orange markers and exiting at purple markers. (Bottom) The corresponding global orthogonal sections $\mathcal{M}^{(l)}_C$ foliating each layer, evaluated at $C \in \{-3, 0, 3\}$ (translucent colored surfaces). Within the transitional layers $\mathcal{L}_1$ and $\mathcal{L}_2$, the localized orthogonal manifolds are structurally glued at the hinges to maintain exact $C^1$ smoothness.}
 \label{fig:3d_cases}
 \end{figure*}
 
 \begin{figure*}[htbp]
     \centering
    
     \includegraphics[width=0.24\textwidth]{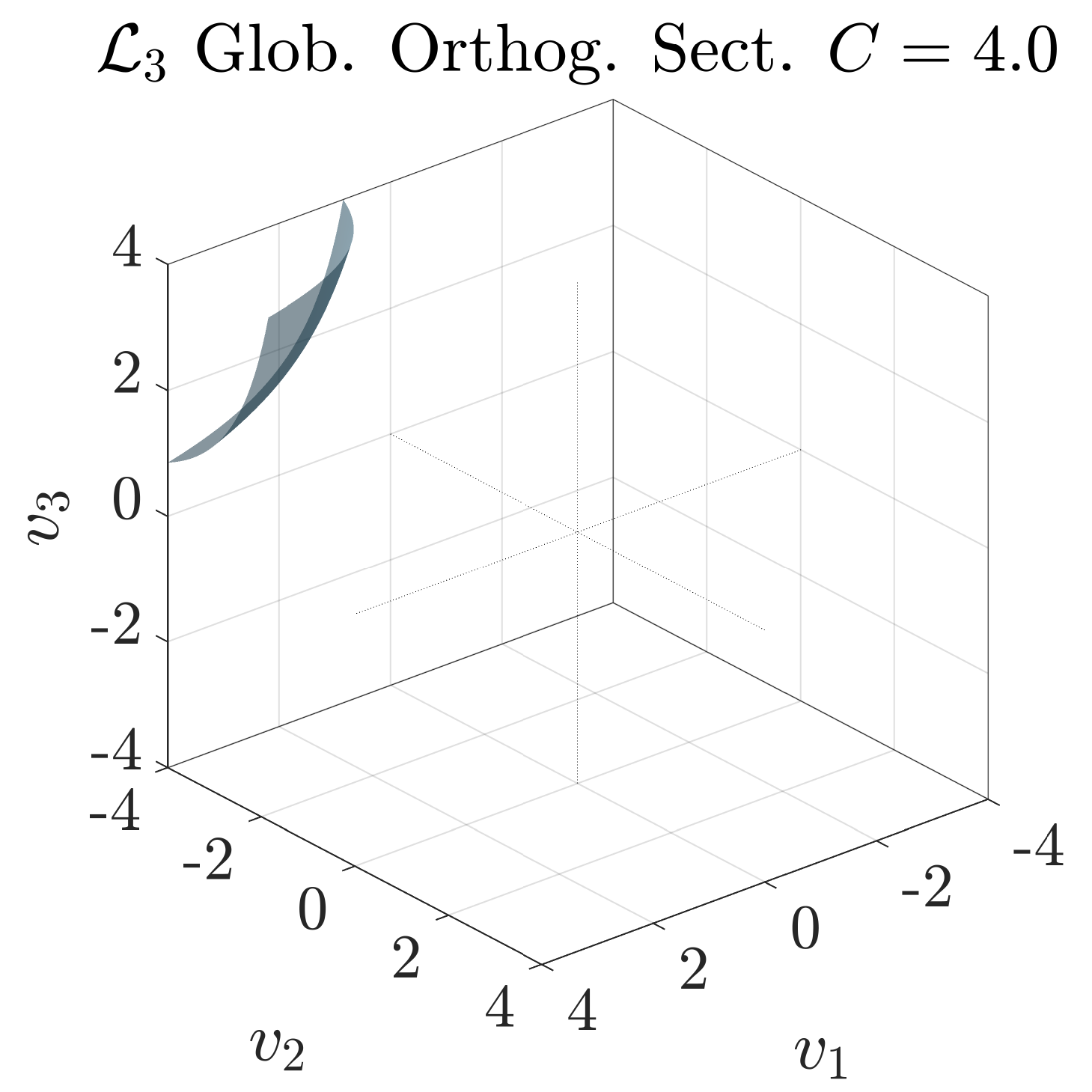}\hfill
     \includegraphics[width=0.24\textwidth]{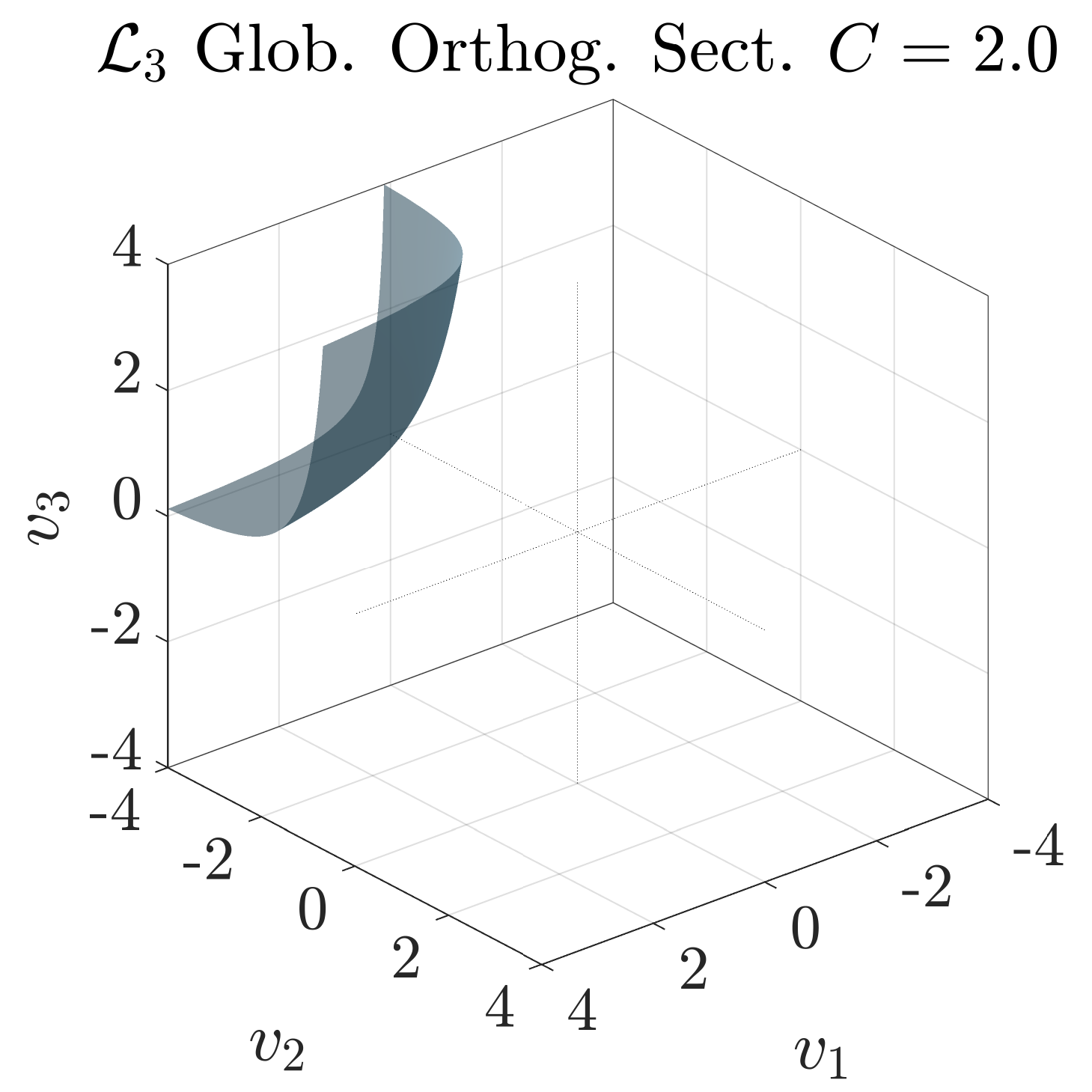}\hfill
     \includegraphics[width=0.24\textwidth]{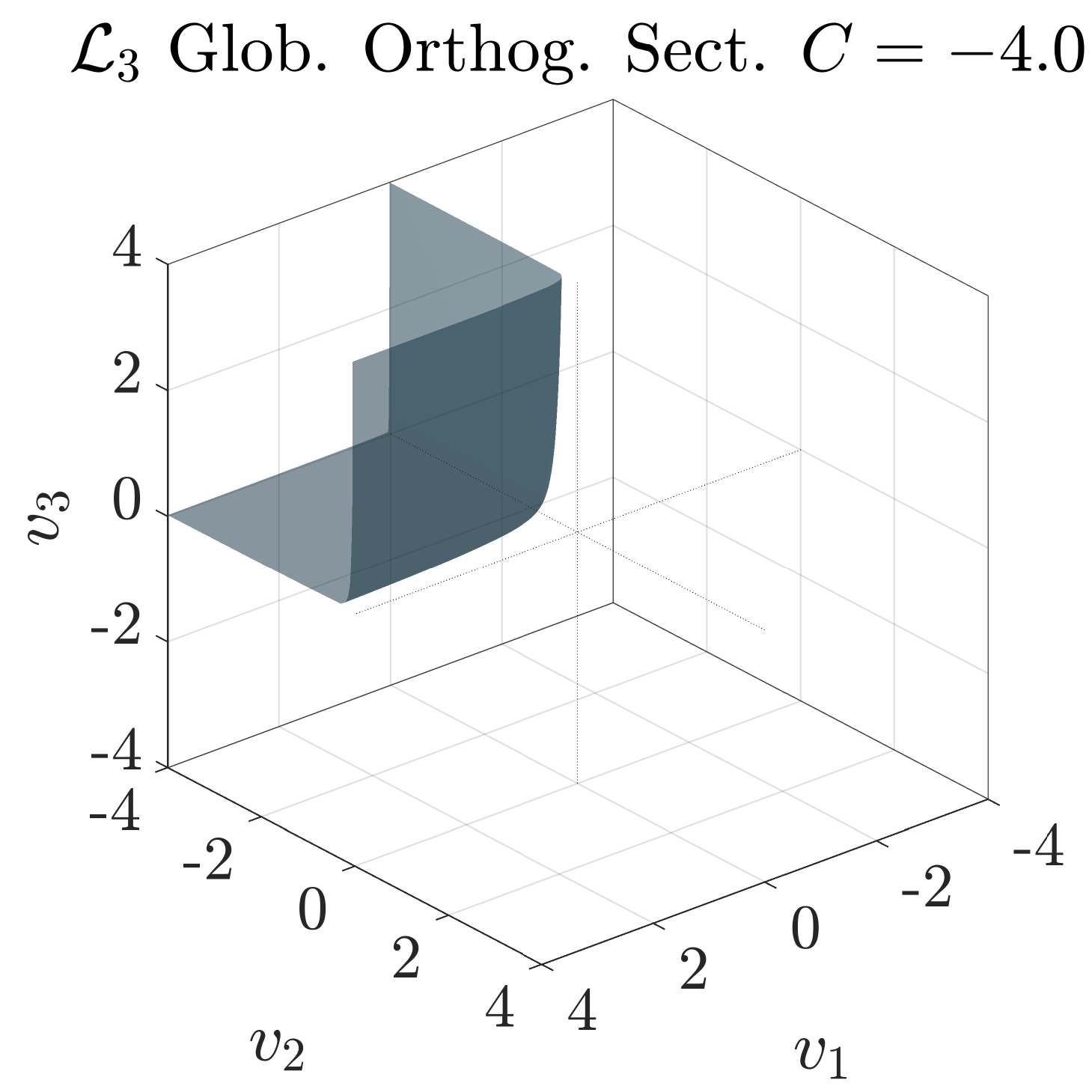}\hfill
     \includegraphics[width=0.24\textwidth]{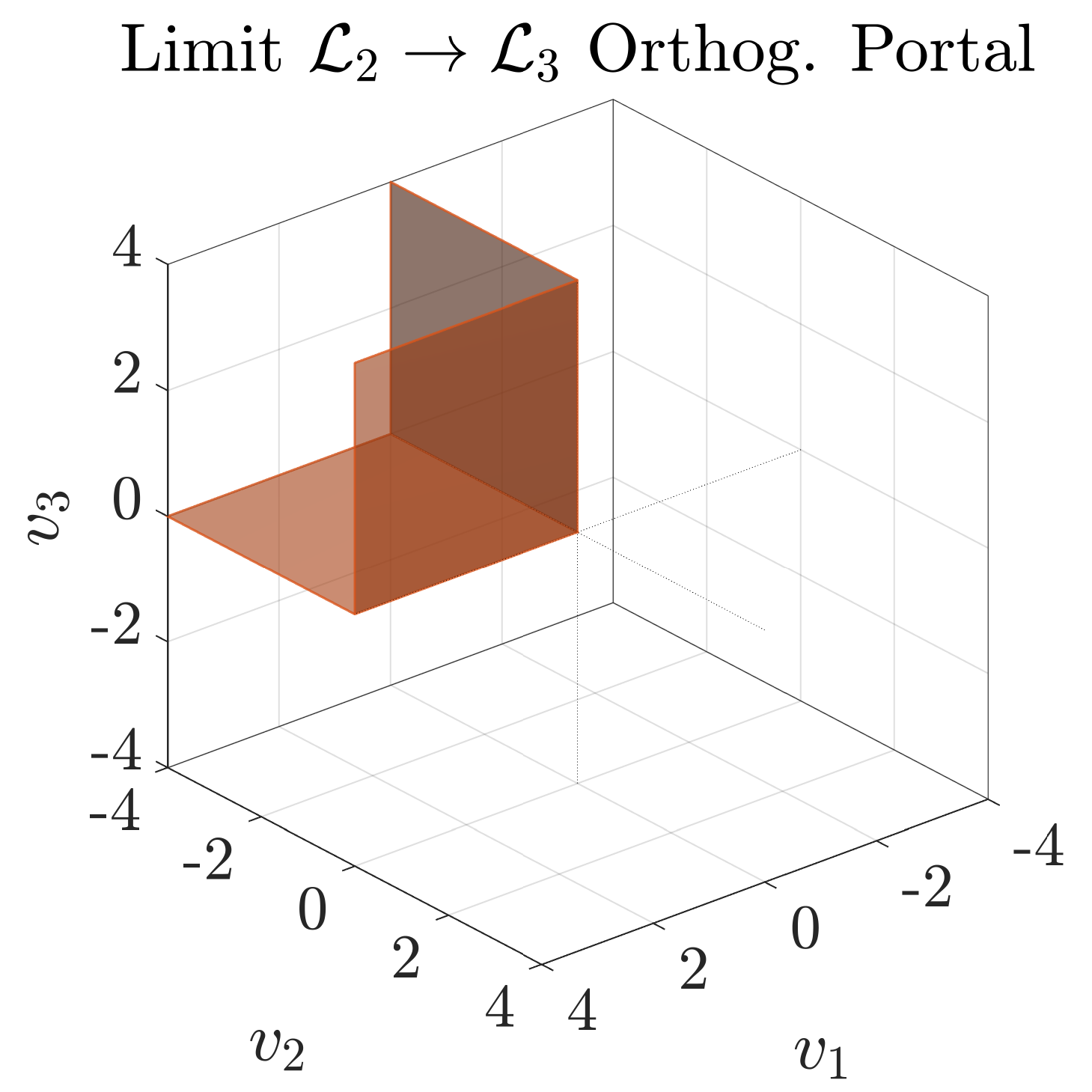}\\[1ex]
    
     \includegraphics[width=0.24\textwidth]{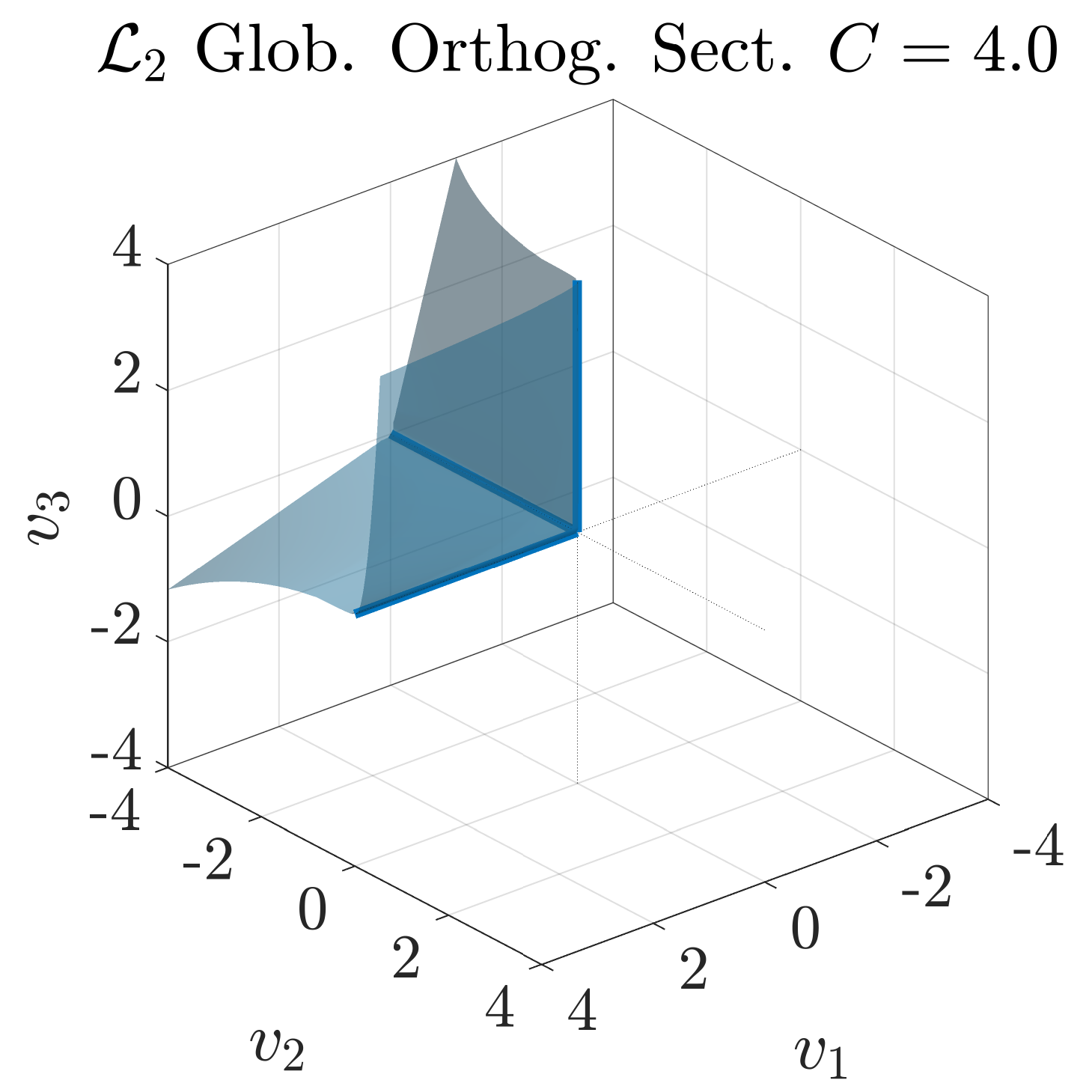}\hfill
     \includegraphics[width=0.24\textwidth]{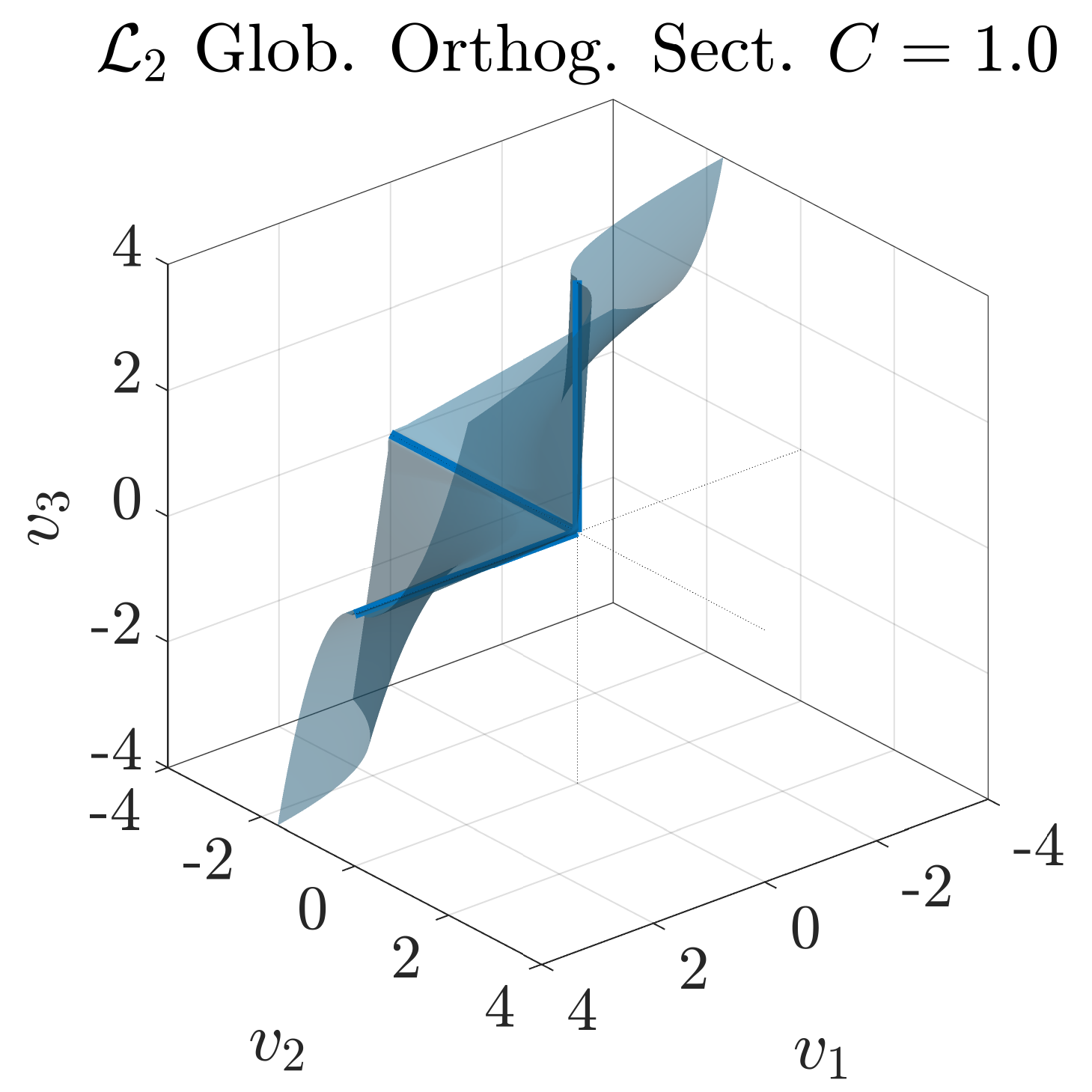}\hfill
     \includegraphics[width=0.24\textwidth]{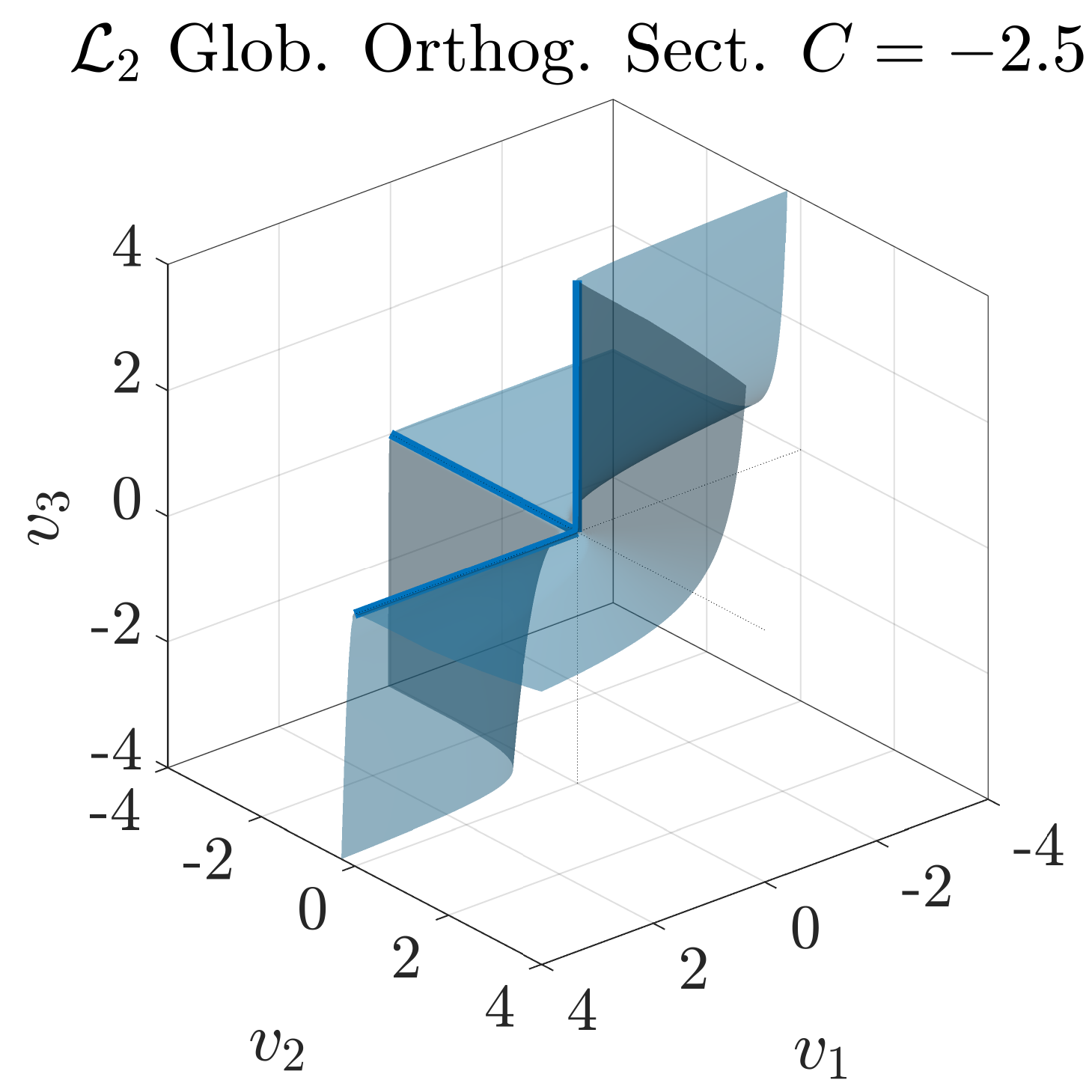}\hfill
     \includegraphics[width=0.24\textwidth]{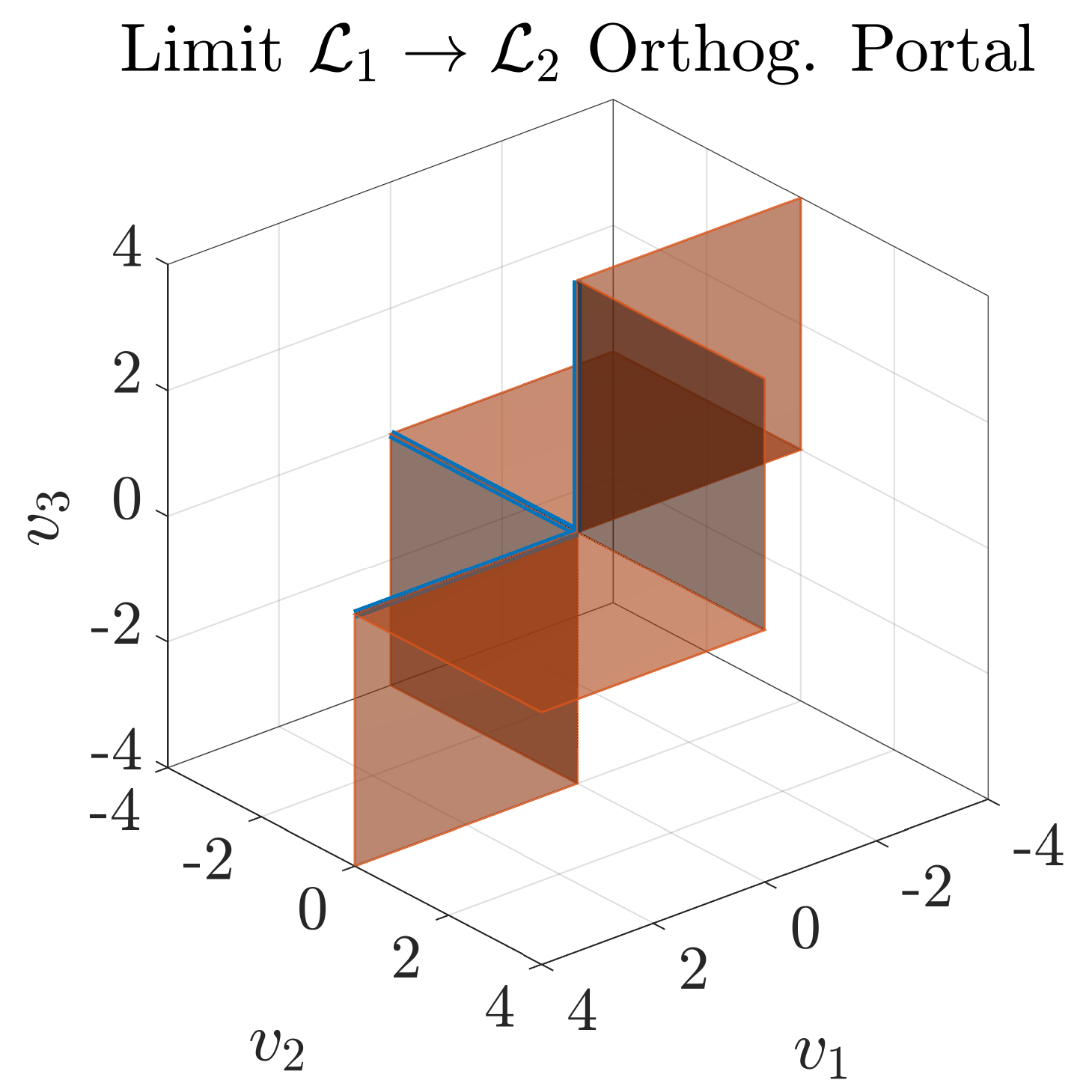}\\[1ex]
    
     \includegraphics[width=0.24\textwidth]{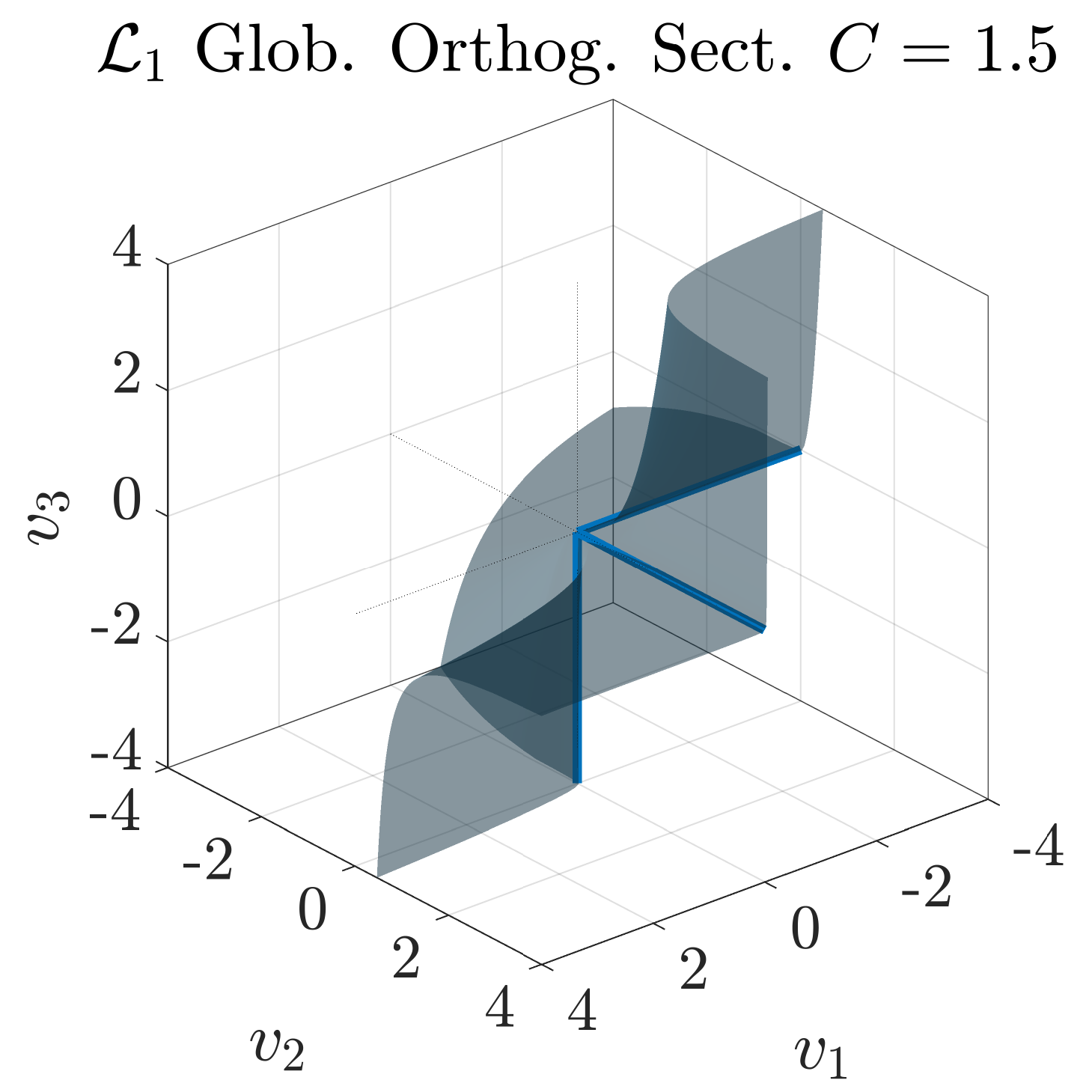}\hfill
     \includegraphics[width=0.24\textwidth]{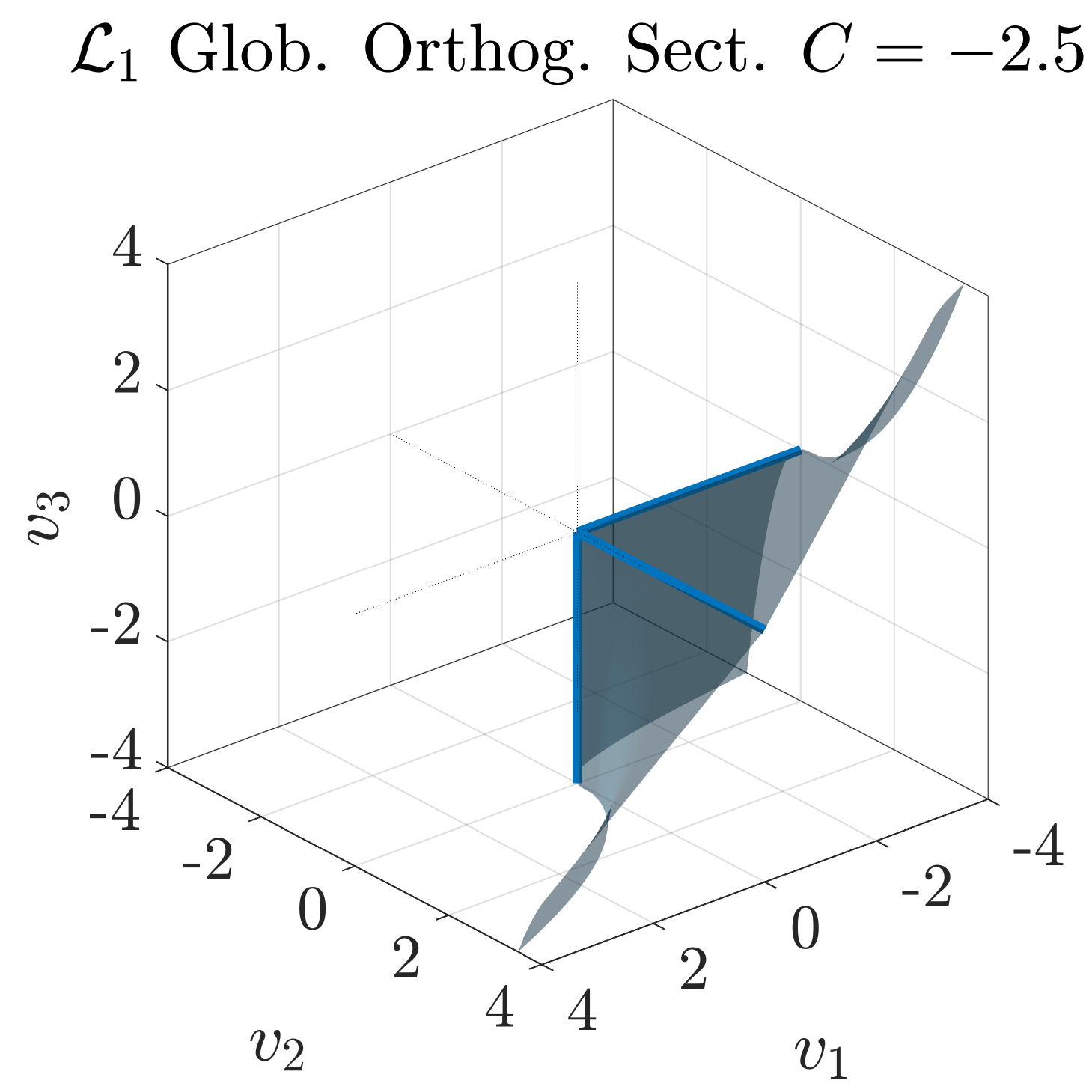}\hfill
     \includegraphics[width=0.24\textwidth]{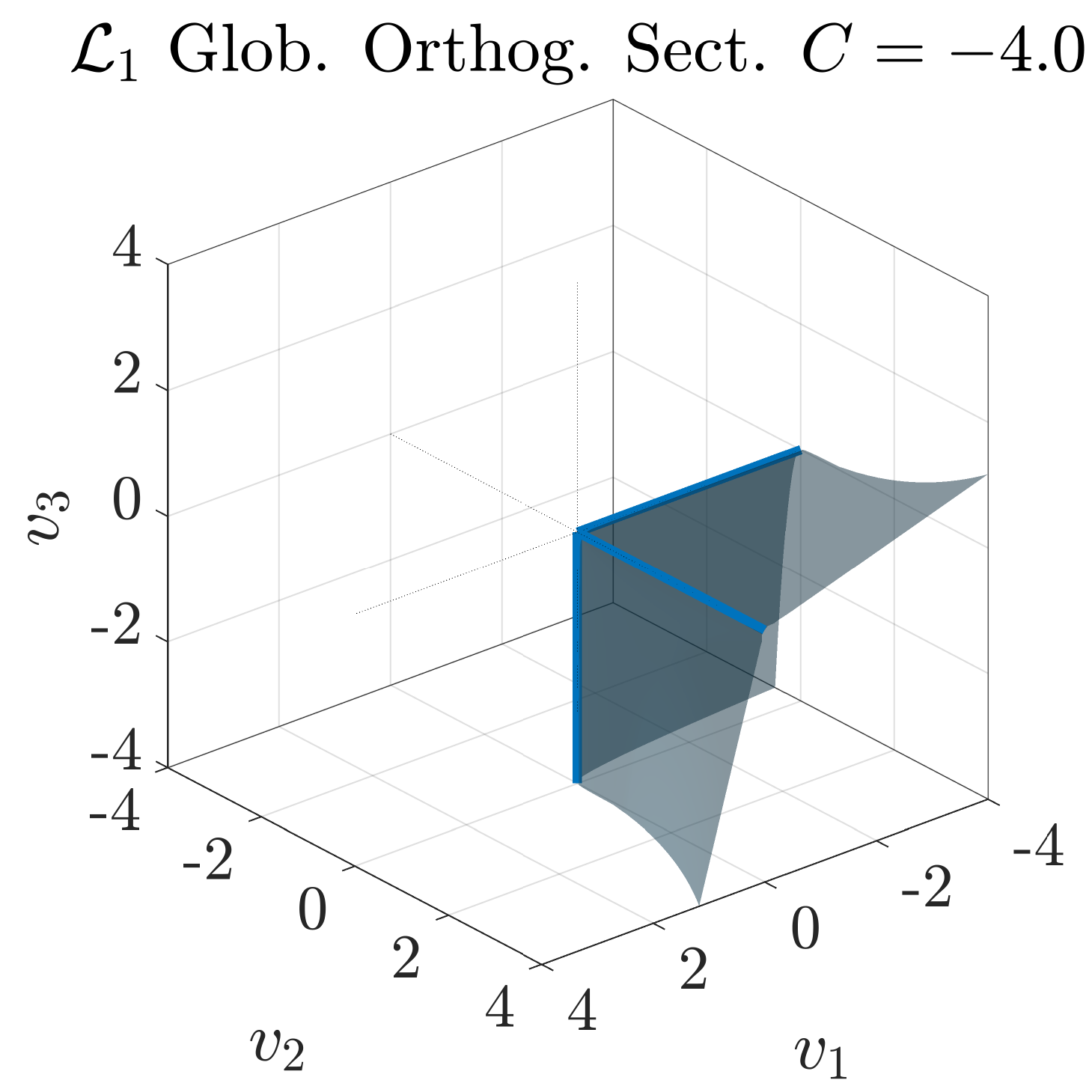}\hfill
     \includegraphics[width=0.24\textwidth]{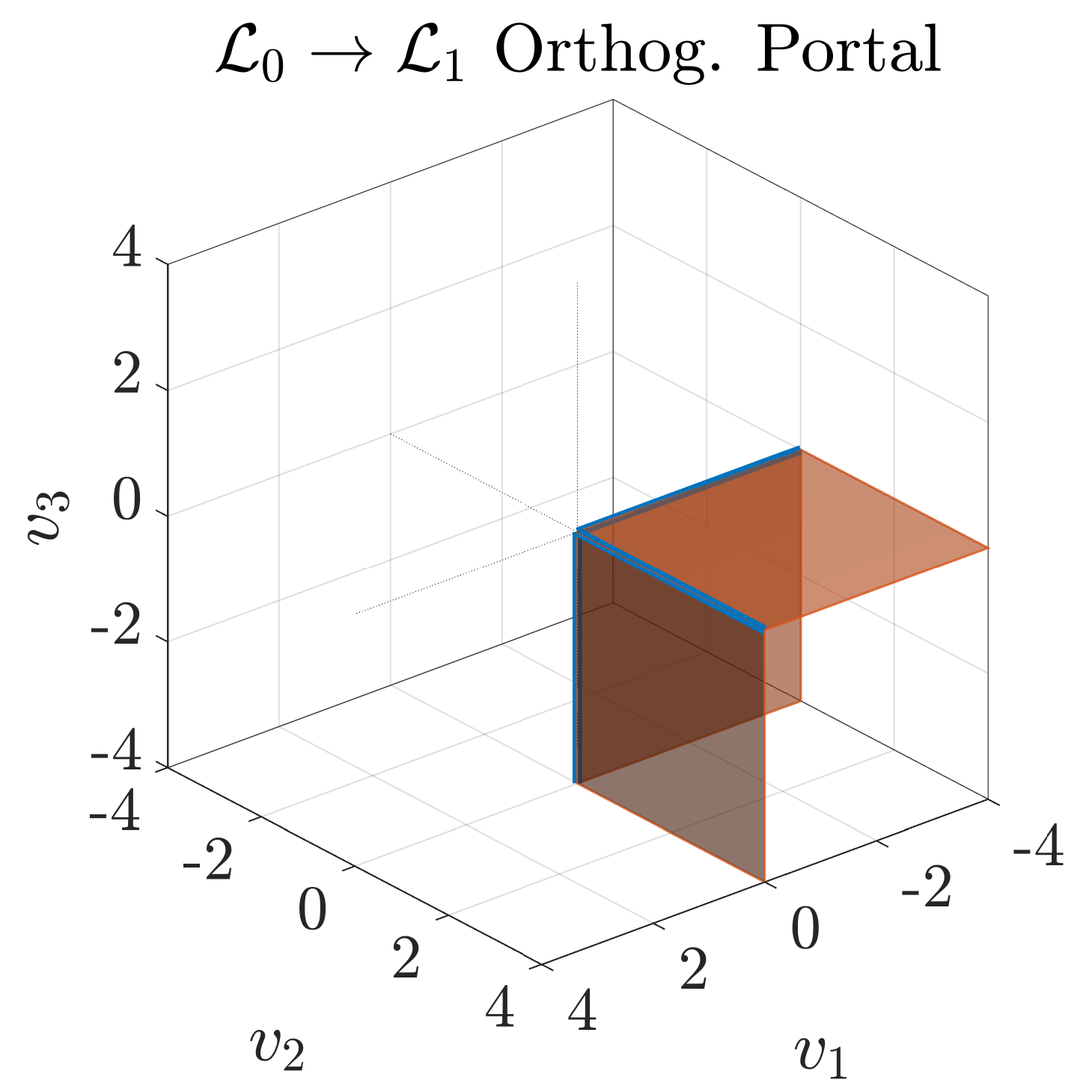}\\[1ex]
    
     \includegraphics[width=0.24\textwidth]{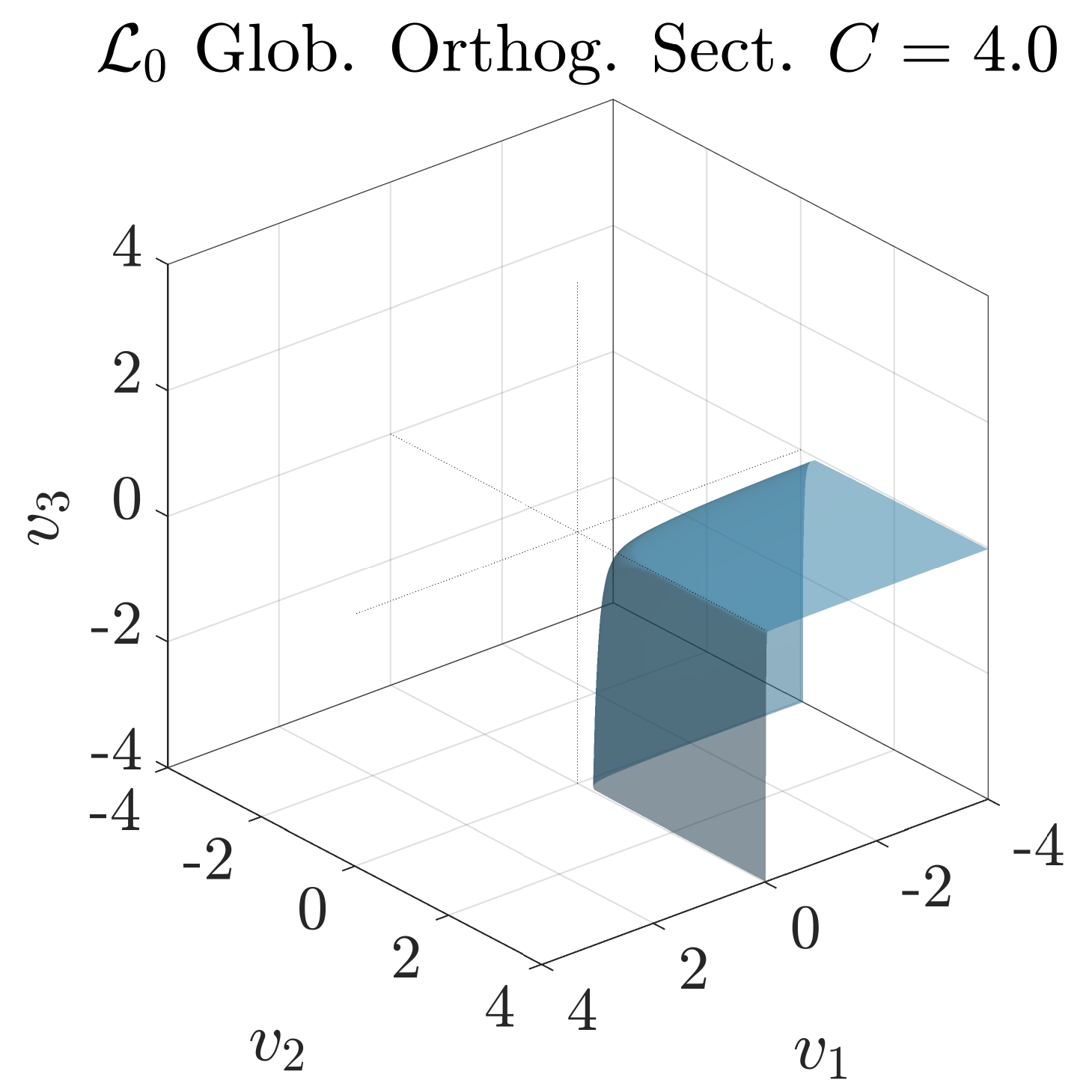}\hfill
     \includegraphics[width=0.24\textwidth]{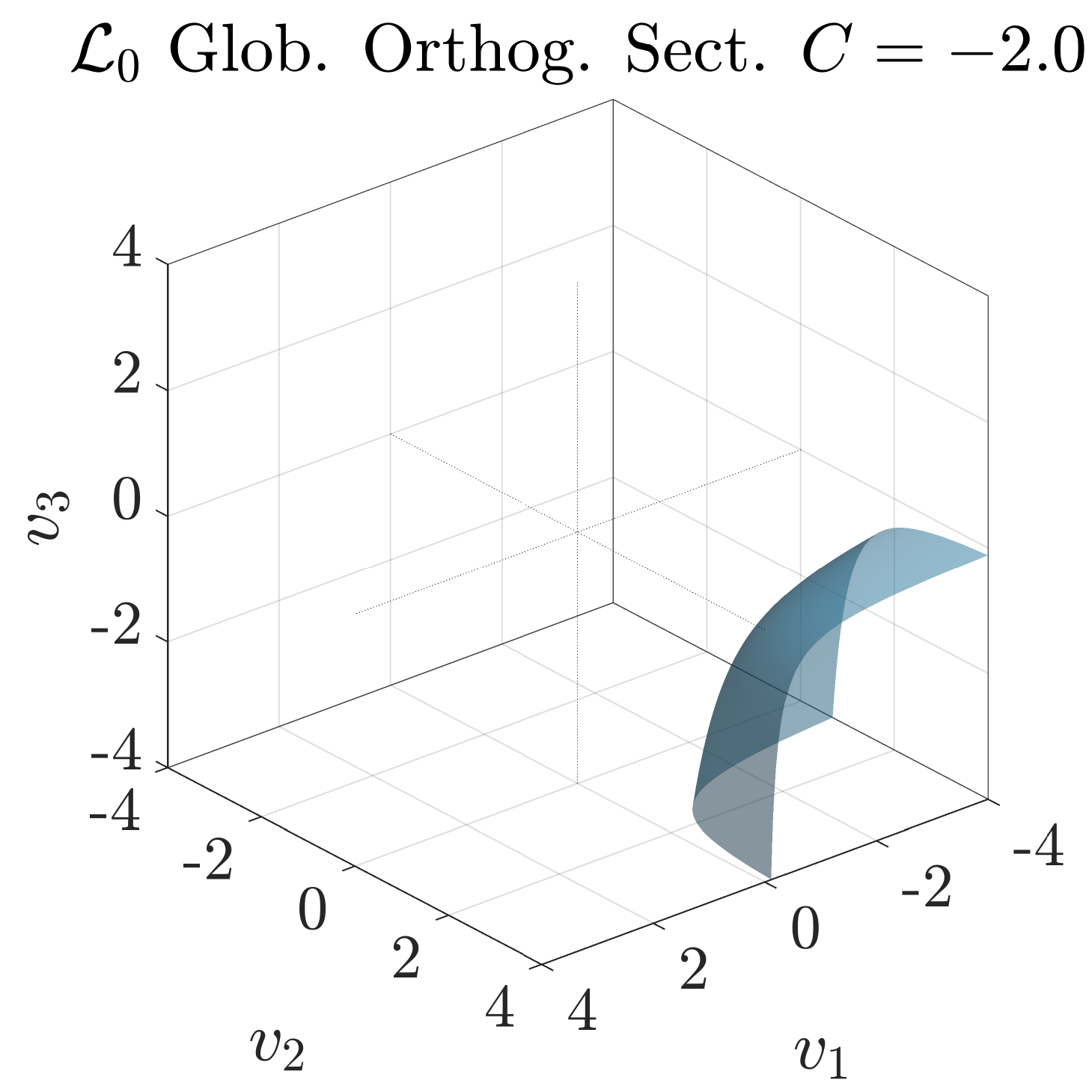}\hfill
     \includegraphics[width=0.24\textwidth]{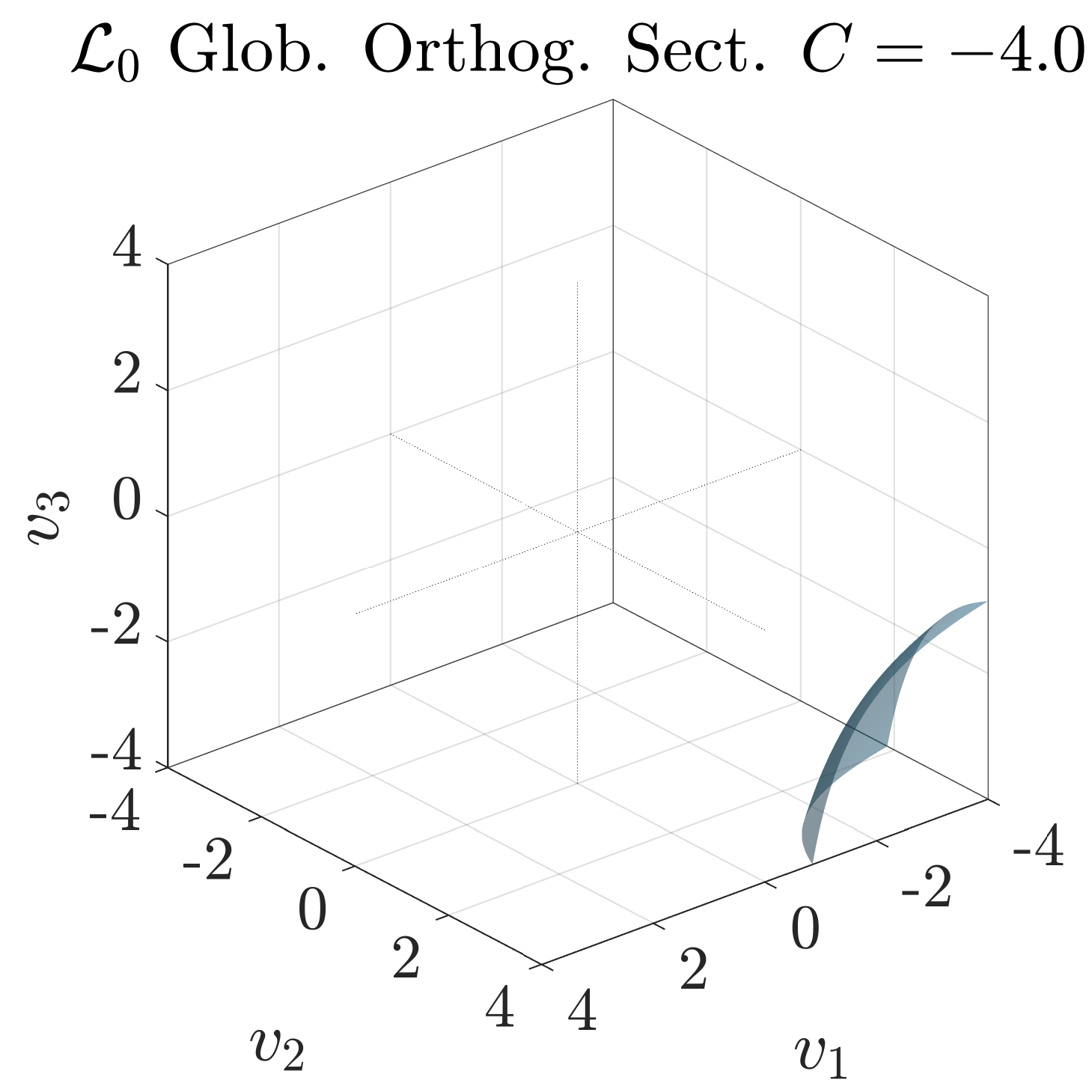}\hfill
     \makebox[0.24\textwidth]{} 
    
     \caption{Sequential progression of the global orthogonal sections already displayed in Fig.~\ref{fig:3d_cases} but this time displayed separately for better visualization. Displayed in reverse order from the positive external orthant  ($\mathcal{L}_3$, top left) back to the negative extremal orthant ($\mathcal{L}_0$, bottom row). Passing through all the intermediate layers and the portals between them.}
     \label{fig:wavefront_sequence}
 \end{figure*}

 \section{Summary: Global Orthogonal Foliation of the Kinetic Space}
 \label{sec:summary_foliation}

 The geometric structure detailed in the preceding analysis can be characterized in the framework of differential geometry as a global orthogonal foliation, formed by exact orthogonal sections of a fiber bundle. We recap and summarize these theoretical results into a list of geometric properties.

 Let $f: \mathcal{V} \to \mathcal{W}$ denote the nonlinear signed-quadratic mapping~\eqref{eq:wrench_map} from the punctured kinetic space $\mathcal{V} \subseteq \mathbb{R}^n \setminus \{0\}$ to the task space $\mathcal{W} \cong \mathbb{R}^m$. The preimages of the task variables define a family of one-dimensional fibers:
 \begin{equation}
     f^{-1}(w) = \{ v \in \mathcal{V} \mid f(v) = w \}.
 \end{equation}

 The analysis establishes the existence of a family of $m$-dimensional orthogonal manifolds, denoted $\mathcal{M}_\kappa$. The index $\kappa = (l, C) \in \mathcal{I}$ uniquely identifies each manifold, where $\mathcal{I} = \{0, \dots, n\} \times \mathbb{R}$. Here, $l$ denotes the manifold's strict confinement to a specific topological layer ($\mathcal{M}_\kappa \subset \mathcal{L}_l$), and $C$ represents its constant potential value ($\Phi(v) = C$). The topological and geometric properties of these manifolds depend strictly on their layer confinement:

 \begin{enumerate}
     \item \emph{Topological Regularity of Extremal Sections:} For orthogonal manifolds constrained entirely within the interior of the extremal layers ($l \in \{0, n\}$, such that $\mathcal{M}_\kappa \subset \mathcal{L}_0^\circ \cup \mathcal{L}_n^\circ$), the orthogonal manifolds never intersect the bounding coordinate hyperplanes. Consequently, the restriction of the map, $f|_{\mathcal{M}_\kappa} : \mathcal{M}_\kappa \to \mathcal{W}$, is a global diffeomorphism. These sections act as exact, globally smooth orthogonal cross-sections of the fiber bundle.

     \item \emph{Singular Boundaries in Transitional Sections:} For orthogonal manifolds spanning the transitional layers ($l \in \{1, \dots, n-1\}$), the sections asymptotically converge toward the origin. Because the origin acts as a topological center but is strictly excluded from $\mathcal{V}$, there is no state $v \in \mathcal{M}_\kappa$ that maps to the origin of $\mathcal{W}$. Therefore, the mapping must be restricted to the punctured task space $\mathcal{W} \setminus \{0\}$. Furthermore, the mapping is structurally partitioned by intersections with a singular set composed of hinges, defined as $\mathcal{H}_\kappa = \{ v \in \mathcal{M}_\kappa \mid \exists i \neq j \text{ s.t. } |v_i|+|v_j| = 0 \}$. These hinges act as lower-dimensional boundaries that connect the disjoint orthant-specific segments (the `petals') of the full manifold. Consequently:
     \begin{itemize}
         \item Topologically, $f|_{\mathcal{M}_\kappa}$ establishes a global homeomorphism between the transitional section $\mathcal{M}_\kappa$ and the punctured task space $\mathcal{W} \setminus \{0\}$.
         \item Differentially, at the hinges $v \in \mathcal{H}_\kappa$, the Jacobian loses full rank, meaning the mapping fails to be a diffeomorphism at these boundaries.
         \item Restricted strictly to the interior of the orthants, the mapping is locally smooth. The manifold domain excluding the hinges, $\mathcal{M}_\kappa \setminus \mathcal{H}_\kappa$, is strictly diffeomorphic to $\mathcal{W} \setminus f(\mathcal{H}_\kappa)$.
         \item As established in Proposition~\ref{prop:hinge_combinatorics}, the closer the transitional layers are to the middle ($l \approx n/2$), the larger the number of hinges. This progressively excises more singular boundaries from the valid mapping, fragmenting the diffeomorphic mapping of $\mathcal{W}$ into an increasingly granular partition of smaller continuous sectors. 
     \end{itemize}
    
     \item \emph{Orthogonal Transversality:} At any non-singular state $v$ ($v \in \mathcal{M}_\kappa$ for extremal sections, and $v \in \mathcal{M}_\kappa \setminus \mathcal{H}_\kappa$ for transitional sections), the tangent space of the manifold is strictly orthogonal to the tangent space of the intersecting fiber. That is, $T_v \mathcal{M}_\kappa \perp T_v \mathcal{F}_{f(v)}$, perfectly satisfying the condition of an orthogonal distribution.
    
     \item \emph{Kinetic Space Foliation:} The continuous family of manifolds partitions the domain $\mathcal{V}$. Every valid kinetic state $v \in \mathcal{V}$ belongs to exactly one manifold $\mathcal{M}_\kappa$. Therefore, the manifolds constitute a complete foliation of the available kinetic space:
     \begin{equation}
         \mathcal{V} = \bigsqcup_{\kappa \in \mathcal{I}} \mathcal{M}_\kappa.
     \end{equation}
 \end{enumerate}

\section{Conclusions}
\label{sec:conclusion}

This work establishes a rigorous differential geometric foundation for redundant systems governed by signed-quadratic actuation maps. Rather than relying on local numerical approximations or algebraic projections, this analysis formalizes the exact native topological structure of the nonlinear kinetic space. 

A primary theoretical result is the proof that the distribution orthogonal to the constant-task fibers is globally integrable. By deriving the exact global logarithmic potential field, we enabled a complete topological classification of the kinetic space. This field foliates the space into $2^n-2$ transitional orthants and $2$ extremal orthants, structurally stratified into $n-1$ transverse layers governed by a strict binomial progression $\binom{n}{l}$. The orthogonal foliation acts as a continuous geometric wavefront, originating in one extremal orthant, traversing the transitional layers via lower-dimensional reciprocal hinges and boundary portals, and terminating in the opposite extremal orthant. 

Exploiting this stratification, we demonstrated that the orthogonal manifolds within the extremal orthants form a global diffeomorphism to the entire unbounded task space. This result proves the theoretical existence of mathematically exact, globally smooth right-inverses that intrinsically guarantee the avoidance of kinematic singularities and geometric rank-loss, fundamentally framing redundancy resolution as a problem of manifold selection rather than constrained numerical optimization. 

While motivated by the physical actuation maps of multirotor and marine vehicles, this topological classification provides a strictly foundational framework for nonlinear systems theory. Future work will leverage these exact topological structures to formulate constrained optimal control problems natively on the manifolds. Specifically, we aim to map physical inequality constraints (actuator saturations) directly onto the bounds of these foliated layers, and to study dynamic flows that optimize kinetic energy strictly along the geometry of the orthogonal distribution, bypassing the structural limitations of Euclidean gradient projections.

\section{Acknowledgements}

The authors thank their colleagues, mentors, and mentees for inspiring discussions. The authors also acknowledge the LLM Gemini 3 (early 2026) for assistance in proofreading and optimizing the data processing and plotting scripts. The functionality of all code and the accuracy of the resulting outputs were manually verified by the authors.

\printbibliography[title={References}]

\begin{IEEEbiography}[{\includegraphics[width=1in,height=1.25in,clip,keepaspectratio]{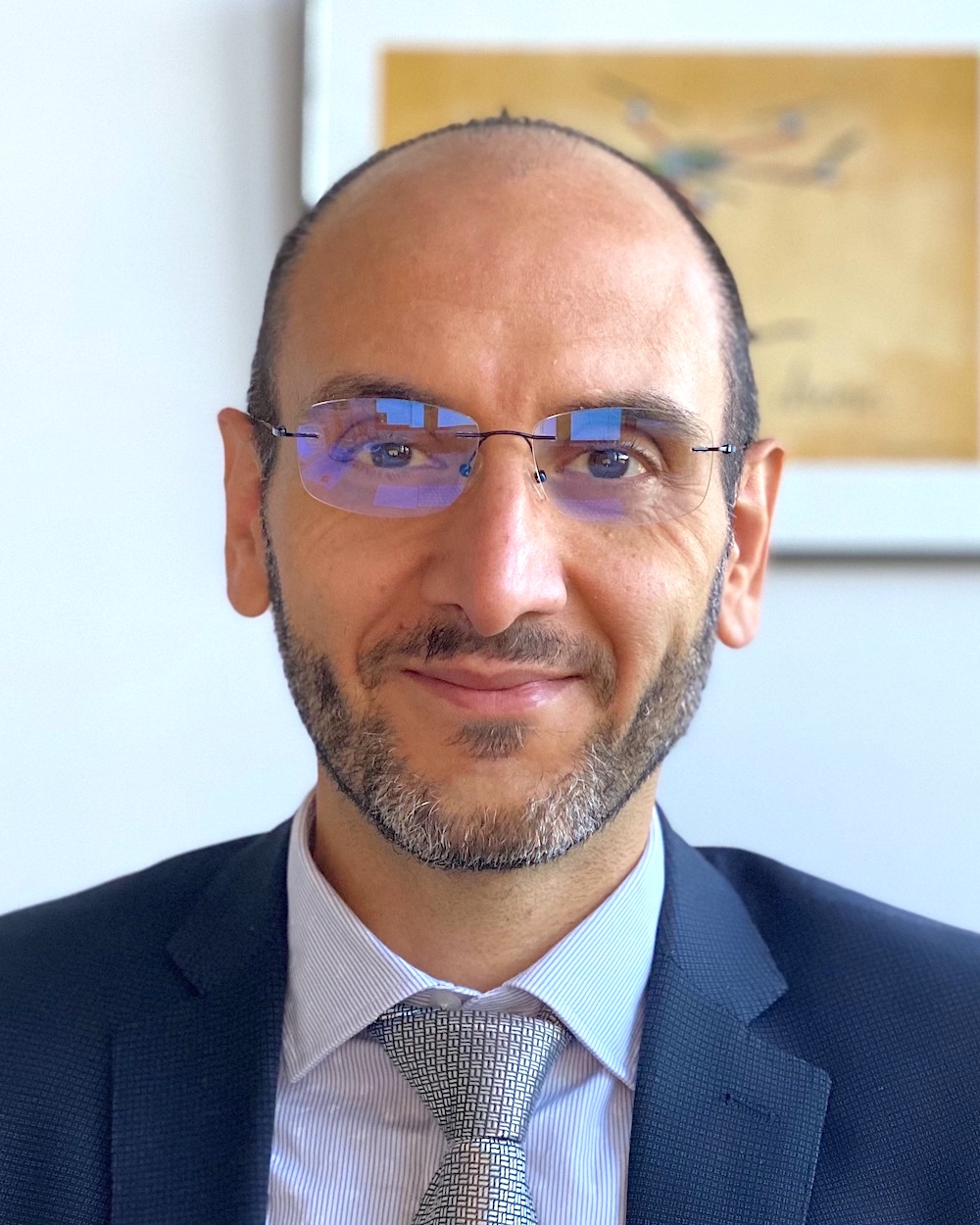}}]{Antonio Franchi} (F'23--SM'16--M'07) received the Ph.D.~degree in system engineering from the Sapienza University of Rome, Rome, Italy in 2010, and the HDR~degree in Science, from the National Polytechnic Institute of Toulouse in 2016. From 2014 to 2019 he was a tenured CNRS researcher at LAAS-CNRS, Toulouse, France. Since 2019 is with the University of Twente. Enschede, The Netherlands where is currently Full Professor in aerial robotics control. Since 2023 he is also Full Professor in the Department of Computer, Control and Management Engineering at Sapienza University of Rome.
He co-authored more than 190 publications in peer-reviewed international journals, books, and conferences on control, estimation, and  design of mechatronic systems, with a particular focus on multi-robot systems and aerial robots.
\end{IEEEbiography}

\end{document}